\documentclass[sn-mathphys-num]{sn-jnl}% Math and Physical Sciences Numbered Reference Style 
\usepackage[symbol]{footmisc}

\usepackage{graphicx}%
\usepackage{multirow}%
\usepackage{amsmath,amssymb,amsfonts}%
\usepackage{amsthm}%
\usepackage{mathrsfs}%
\usepackage[title]{appendix}%
\usepackage{xcolor}%
\usepackage{textcomp}%
\usepackage{manyfoot}%
\usepackage{booktabs}%
\usepackage{algorithm}%
\usepackage{algorithmicx}%
\usepackage{algpseudocode}%
\usepackage{listings}%
\usepackage{tikz}
\usepackage{placeins}
\usetikzlibrary{positioning}
\tikzset{
    line/.style={draw, very thick, color=black!50},
    block/.style={rectangle, draw, text centered, node distance=4em},
    onslide/.code args={<#1>#2}{\only<#1>{\pgfkeysalso{#2}}}, 
  }
\usepackage{tkz-graph}
\usepackage{pgfplots}
\usepackage{blkarray,bigstrut}%
\usetikzlibrary{intersections}
\usetikzlibrary{patterns}
\pgfplotsset{compat=1.17}
\usepackage{multirow}
\usepackage{array}
\usepackage{newfloat}
\usepackage{mathtools}

\theoremstyle{thmstyleone}%
\newtheorem{theorem}{Theorem}%  meant for continuous numbers
\newtheorem{observation}[theorem]{Observation}
\newtheorem{proposition}[theorem]{Proposition}% 
\newtheorem{corollary}[theorem]{Corollary}
\newtheorem{lemma}[theorem]{Lemma}

\newtheorem{example}[theorem]{Example}%

\newtheorem{definition}[theorem]{Definition}%

\newtheorem{problem}[theorem]{Problem}
\newtheorem{alg}[theorem]{Algorithm}

\newcommand{\N}{\mathbb{N}}

\newcommand{\R}{\mathbb{R}}
\newcommand{\Z}{\mathbb{Z}}

\let\cline\cmidrule

\begin{document}

\title[Article Title]{NP-Completeness of the Combinatorial Distance Matrix Realisation Problem}


\author*[1,3]{David L. Fairbairn}\email{david.fairbairn@tharsus.co.uk}

\author[2]{George B. Mertzios\footnote[2]{Supported by the EPSRC grant EP/P020372/1.} \ }\email{george.mertzios@durham.ac.uk}

\author[1]{Norbert Peyerimhoff}\email{norbert.peyerimhoff@durham.ac.uk}

% Department of Mathematical Sciences, Durham University, Upper Mountjoy Campus, Stockton Rd, Durham DH1 3LE, England, UK
\affil[1]{\orgdiv{Department of Mathematical Sciences}, \orgname{Durham University},  \country{UK}}

% Department of Computer Science, Durham University, South Road, Durham DH1 3LE, England, UK
\affil[2]{\orgdiv{Department of Computer Science}, \orgname{Durham University}, \country{UK}}

% Tharsus, Coniston Rd, Blyth NE24 4RF, England, UK
\affil[3]{\orgname{Tharsus Limited}, \city{Blyth}, \state{Northumberland}, \country{UK}}


\abstract{
    The \(k\)-\textsc{CombDMR} problem is that of determining whether an \(n \times n\) distance matrix can be realised by \(n\) vertices in some undirected graph with \(n + k\) vertices.
    This problem has a simple solution in the case \(k=0\).
    In this paper we show that this problem is polynomial time solvable for \(k=1\) and \(k=2\).
    Moreover, we provide algorithms to construct such graph realisations by solving appropriate 2-SAT instances.
    In the case where \(k \geq 3\), this problem is NP-complete.
    We show this by a reduction of the \(k\)-colourability problem to the \(k\)-\textsc{CombDMR} problem.
    Finally, we discuss the simpler polynomial time solvable problem of tree realisability for a given distance matrix.
}

\keywords{Distance matrix, graph realisation, NP-completeness, polynomial time algorithm, graph colourability.\vspace{-0.5cm}}



\maketitle

\section{Introduction}\label{sec:intro}

This paper is concerned with the problem of computing combinatorial graph realisations with \(n + k\) vertices, for an \(n \times n\) integer valued distance matrix with a prescribed number \(k\) of additional vertices.
Our natural minimality criterion is to find a graph realisation with the smallest number of vertices.

Graph realisation problems have many practical applications, for example, they appear prominently within the fields of Phylogenetics and Evolutionary Trees \cite{Semple2003} and Network Tomography \cite{Chung2001,herman2012discrete}. For a survey we refer the reader to See Bar-Noy et al. 2021 \cite{Barnoy2021} and references therein.
Various kinds of graph realisation problems have been studied in the literature, most of them are concerned with weighted graphs with a different optimisation criterion, namely, \emph{minimising the sum of the edge weights}.
This problem was first introduced by Hakimi and Yau 1965 \cite{Hakimi1965DistanceMO}.
Amongst their results are a set of necessary and sufficient conditions for realisability of a given matrix and a proof of uniqueness of shortest length tree realisations.
The existence of a weighted graph realisation with minimum total edge weight for any given distance matrix was shown by Dress 1984 \cite{Dress1984}.
Moreover, he proved the existence of an optimum solution with at most \(n^4\) vertices for any \(n \times n\) distance matrix. In his result, the vertices are only the branch points (vertices with \(\ge 3\) incident edges) or leaves (vertices with just one incident edges), since all vertices with precisely \(2\) incident edges can be condensed.
Finding weighted graph realisations having the smallest sum of edge weights is NP-hard. More specifically, Alth\"{o}fer 1988 \cite{Althofer1988} proved that this problem remains NP-hard even in the case where the input distance matrix has integer values (while the edge weights are still real valued). 
Alth\"{o}fer also showed that in the case of integer valued distance matrices, there is always an optimum realisation with rational edge weights \cite{Althofer1988}.
Chung, Garrett and Graham 2001 \cite{Chung2001} considered a weak version of the weighted graph realisation problem, namely, finding optimum graph realisations for which the distance matrix provides a lower bound on the distances of the corresponding \(n\) vertices.
They showed that even this weak version of the problem is NP-hard \cite{Chung2001}.

This paper's focus of finding combinatorial graph realisations for a prescribed integer valued distance matrix with a prescribed number of additional vertices is naturally motivated by any domain where vertices are of limited supply (or expensive) and edges are of inconsequential cost.
% Our motivation comes from the field of Multi-Agent Path-Finding (MAPF) \cite{Stern2019,Sharon2015} which is the problem of finding a set of collision-free paths for a group of \(n\) agents from their unique start locations to their respective unique goal locations within some graph, which is known to be NP-hard to solve optimally~\cite{Sharon2015}.
% Taking the pairwise distances between start and goal locations is sufficient information to quickly compute a feasible solution to the MAPF problem in some cases (see Atzmon et al. 2023 \cite{Atzmon2023}) irrespective of the underlying topology of the graph.
% A heuristically difficult instance of the MAPF problem is one on a combinatorial graph with the smallest number of vertices.
% We could consider the near-optimality of solutions to the MAPF problem across combinatorial graph realisations for a prescribed integer valued distance matrix (given by the pairwise distances between start and goal locations of the agents) with a minimum number of additional vertices.
% Such solutions could provide an indication of the difficulty or bound the solution cost of a MAPF instance on an arbitrary graph with the same pairwise distances.
As we show in this paper, the problem of finding combinatorial graph realisations is NP-complete for a prescribed number of~3 or more additional vertices (see Theorem \ref{thm:maink}), while it is polynomial time solvable for \(0\), \(1\) or \(2\) additional vertices (see Theorem \ref{thm:nequalk}, Theorem \ref{thm:kp1} and Theorem \ref{thm:kp2} respectively).

In this paper we use the notation \([n] = \{1,2,\dots,n\}\) for any \(n \in \N\) and denote the set of all non-negative integers by \(\N_0\) (that is \(\N_0 = \N \cup \{0\}\)).
First we introduce the following problem for every integer $k\in \mathbb{N}_0$.
\medskip

\begin{problem}
    \label{prob:drp}
    \textsc{$k$-Combinatorial Distance Matrix Realisation Problem \\($k$-CombDMR)}\\
    \emph{Input:} An $n \times n$ matrix $D$ with non-negative integer values. \\    
    \emph{Question:} Does there exist a simple (unweighted) graph $G=(V,E)$ with $|V| \leq n+k$ and an injective mapping $\Phi: [n] \rightarrow V$ such that the shortest-path distance function $d$ in $G$ satisfies
    \begin{equation*}
        \label{eq:drp-new}
        d(\Phi(i), \Phi(j)) = D_{ij}
    \end{equation*}
        for all \(i,j \in [n]\)?
\end{problem}

\medskip

We call such a pair \((G, \Phi)\) a {\emph{graph realisation}} of \(D\).
Given an $n\times n$ matrix $D$, any graph realisation \((G,\Phi)\) which has the smallest number of vertices is called a {\emph{minimum graph realisation}} of \(D\).

\begin{example}
    \label{ex:drp}
    Two possible graph realisations of the following matrix \(D\), are given in Figure \ref{fig:graphrealisations} with $n=3$, while $k=3$ and $k=1$, respectively.

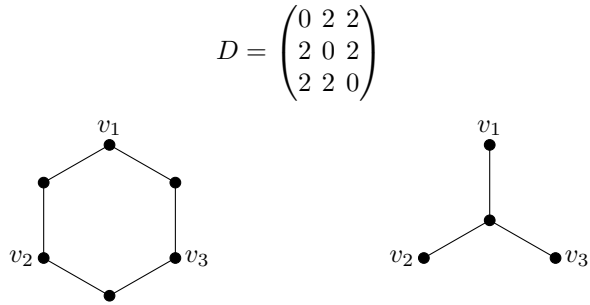
\begin{figure}[h!]
    \begin{equation*}
        \label{eq:drpex}
        D = \begin{pmatrix}
           0 & 2 & 2 \\
           2 & 0 & 2 \\
           2 & 2 & 0 
        \end{pmatrix}
    \end{equation*}
    \begin{center}
\begin{tikzpicture}[scale=1]

    \draw (30:1) -- (90:1) -- (150:1) -- (210:1) -- (270:1) -- (330:1) -- cycle;

    % Placing the points on the hexagon
    \filldraw[black] (30:1) circle (2pt) ;
    \filldraw[black] (90:1) circle (2pt) node[anchor=south] {$v_1$};
    \filldraw[black] (150:1) circle (2pt) ;
    \filldraw[black] (210:1) circle (2pt) node[anchor=east] {$v_2$};
    \filldraw[black] (270:1) circle (2pt) ;
    \filldraw[black] (330:1) circle (2pt) node[anchor=west] {$v_3$};

    % Second diagram - a Y shape with points and lines
    \begin{scope}[xshift=5cm]
    \draw (0,0) -- (90:1);
    \draw (0,0) -- (210:1);
    \draw (0,0) -- (330:1);
    \filldraw[black] (0,0) circle (2pt) node[anchor=north] {};
    \filldraw[black] (90:1) circle (2pt) node[anchor=south] {$v_1$};
    \filldraw[black] (210:1) circle (2pt) node[anchor=east] {$v_2$};
    \filldraw[black] (330:1) circle (2pt) node[anchor=west] {$v_3$};
    \end{scope}
    \end{tikzpicture}
    \caption{Two graph realisations of the above matrix \(D\), where \(\Phi(i) = v_i\) for \(i \in [3]\), while $k=3$ and $k=1$ in the left and the right realisation, respectively. The right realisation is a \emph{minimum}.}
    \label{fig:graphrealisations}
\end{center}
\end{figure}

\end{example}

\medskip

It is important to note that \(k\)-\textsc{CombDMR} is distinct from the weighted graph realisation problem.
Within the weighted graph realisation problem the edges are equipped with positive real valued weights (their lengths) with the aim to minimise the sum of the edge weights of the graph realisation, whereas in \(k\)-\textsc{CombDMR} we are only concerned with minimising the number of vertices in the graph realisation.
Take for instance the distance matrix \(D\) in \eqref{eq:drpcomp1} and optimum solutions within these two problems as shown in Figure \ref{fig:graphrealisationscomp1}.
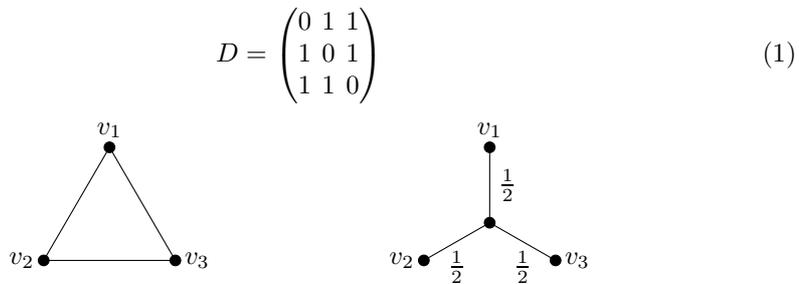
\begin{figure}[h!]
    \begin{equation}
        \label{eq:drpcomp1}
        D = \begin{pmatrix}
           0 & 1 & 1 \\
           1 & 0 & 1 \\
           1 & 1 & 0 
        \end{pmatrix}
    \end{equation}
    \begin{center}
\begin{tikzpicture}[scale=1]

    \draw (90:1) -- (210:1) -- (330:1) -- cycle;

    % Placing the points on the triangle
    \filldraw[black] (90:1) circle (2pt) node[anchor=south] {$v_1$};
    \filldraw[black] (210:1) circle (2pt) node[anchor=east] {$v_2$};
    \filldraw[black] (330:1) circle (2pt) node[anchor=west] {$v_3$};

    % Second diagram - a Y shape with points and lines
    \begin{scope}[xshift=5cm]
    \draw (0,0) -- node[right] {$\frac{1}{2}$} (90:1);
    \draw (0,0) -- node[below] {$\frac{1}{2}$} (210:1);    
    \draw (0,0) -- node[below] {$\frac{1}{2}$} (330:1);
    \filldraw[black] (0,0) circle (2pt) node[anchor=north] {};
    \filldraw[black] (90:1) circle (2pt) node[anchor=south] {$v_1$};
    \filldraw[black] (210:1) circle (2pt) node[anchor=east] {$v_2$};
    \filldraw[black] (330:1) circle (2pt) node[anchor=west] {$v_3$};
    \end{scope}
    \end{tikzpicture}
    \caption{Minimum graph realisations of \(D\) as in \eqref{eq:drpcomp1}, for \(k\)-\textsc{CombDMR} (Left) and for the weighted graph realisation problem (Right).}
    \label{fig:graphrealisationscomp1}
\end{center}
\end{figure}

Figure \ref{fig:graphrealisationscomp1} shows that optimum solutions for \(k\)-\textsc{CombDMR} and for the weighted graph realisation problem can differ significantly. 
An optimum solution for \(k\)-\textsc{CombDMR} may not be trivially transformed into an optimum solution for the weighted graph realisation problem, and vice versa.
The reader may ask what happens if we consider the weighted graph realisation problem with the additional constraint that the weights must be integers -- could we transform any optimum solution for this problem into an optimum solution for \(k\)-\textsc{CombDMR} by replacing weighted edges with paths of length equal to the weight?
Figure \ref{fig:graphrealisationscomp2} shows that this is not always the case, as the weighted graph realisation problem with integer weights may have multiple solutions, some of which do not have a corresponding optimum solution for \(k\)-\textsc{CombDMR} under this transformation.
Therefore, \(k\)-\textsc{CombDMR} is also a distinct problem from the weighted graph realisation problem with integer weights.

\begin{figure}[h!]
    \begin{equation}
        \label{eq:drpcomp2}
        D = \begin{pmatrix}
        0 & 2 & 2 & 2 & 1 & 3 & 3 & 1 \\
        2 & 0 & 2 & 2 & 3 & 1 & 3 & 1 \\
        2 & 2 & 0 & 2 & 3 & 1 & 1 & 3 \\
        2 & 2 & 2 & 0 & 1 & 3 & 1 & 3 \\
        1 & 3 & 3 & 1 & 0 & 4 & 2 & 2 \\
        3 & 1 & 1 & 3 & 4 & 0 & 2 & 2 \\    
        3 & 3 & 1 & 1 & 2 & 2 & 0 & 4 \\
        1 & 1 & 3 & 3 & 2 & 2 & 4 & 0
        \end{pmatrix}
    \end{equation}
    \begin{center}
\begin{tikzpicture}[scale=1]
    
        % 1->_->3-> 7 -> 4 -> _ -> 2 -> 8 cycle
        \draw (1,3) -- (2,3) -- (3,3) -- (3,2) -- (3,1) -- (2,1) -- (1,1) -- (1,2) -- cycle;
        % 2 -> 6 -> 3 not cycle
        \draw (1,1) -- (0,4) -- (3,3);
        % 1 -> 5 -> 4 not cycle
        \draw (1,3) -- (0,0) -- (3,1);
    
        \filldraw[black] (1,3) circle (2pt) node[anchor=south] {$v_1$};
        \filldraw[black] (2,3) circle (2pt) node[anchor=south] {};
        \filldraw[black] (3,3) circle (2pt) node[anchor=south] {$v_3$};

        \filldraw[black] (1,1) circle (2pt) node[anchor=north] {$v_2$};
        \filldraw[black] (2,1) circle (2pt) node[anchor=north] {};
        \filldraw[black] (3,1) circle (2pt) node[anchor=north] {$v_4$};
    
        \filldraw[black] (1,2) circle (2pt) node[anchor=west] {$v_8$};
        \filldraw[black] (3,2) circle (2pt) node[anchor=east] {$v_7$};

        \filldraw[black] (0,4) circle (2pt) node[anchor=south] {$v_6$};
        \filldraw[black] (0,0) circle (2pt) node[anchor=north] {$v_5$};

    \begin{scope}[xshift=4cm]
        % 1->_->3-> 7 -> 4 -> _ -> 2 -> 8 cycle
        \draw (1,3) -- (2,2) -- (3,3) -- (3,2) -- (3,1) -- (2,2) -- (1,1) -- (1,2) -- cycle;
        % 2 -> 6 -> 3 not cycle
        \draw (1,1) -- (0,4) -- (3,3);
        % 1 -> 5 -> 4 not cycle
        \draw (1,3) -- (0,0) -- (3,1);
    
        \filldraw[black] (1,3) circle (2pt) node[anchor=south] {$v_1$};
        \filldraw[black] (2,2) circle (2pt) node[anchor=south] {};
        \filldraw[black] (3,3) circle (2pt) node[anchor=south] {$v_3$};

        \filldraw[black] (1,1) circle (2pt) node[anchor=north] {$v_2$};
        \filldraw[black] (3,1) circle (2pt) node[anchor=north] {$v_4$};
    
        \filldraw[black] (1,2) circle (2pt) node[anchor=west] {$v_8$};
        \filldraw[black] (3,2) circle (2pt) node[anchor=east] {$v_7$};

        \filldraw[black] (0,4) circle (2pt) node[anchor=south] {$v_6$};
        \filldraw[black] (0,0) circle (2pt) node[anchor=north] {$v_5$};

    \end{scope}
        \end{tikzpicture}
        \caption{Optimum graph realisations of \(D\) as in \eqref{eq:drpcomp2}, with \(\Phi(i)=v_i, i \in [8]\) for the weighted graph realisation problem with integer weights \(w(e) = 1\) (Left and Right). Only the right graph realisation is optimum for the combinatorial distance realisation problem.}
        \label{fig:graphrealisationscomp2}
    \end{center}
\end{figure}
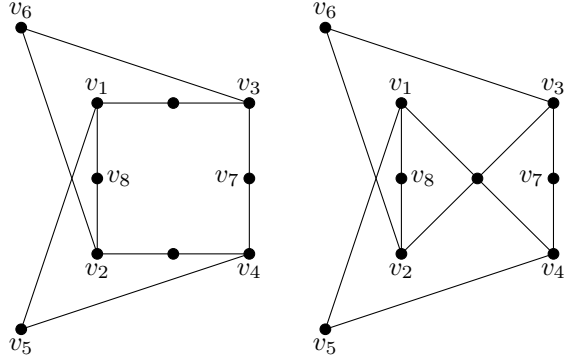
As it is not immediately apparent that the graph realisations in Figure \ref{fig:graphrealisationscomp2} are minimum with respect to their respective optimality criteria, we will now provide a proof of this statement.
\medskip
\begin{lemma}
    For the distance matrix \(D\) in \eqref{eq:drpcomp2}, the minimum sum of integer edge weights for a weighted graph realisation of \(D\) is 12 and the minimum number of vertices for a combinatorial graph realisation of \(D\) is 9.
\end{lemma}
\begin{proof}
    The entries of \(D\) which are equal to 1 necessarily correspond to edges in any graph realisation of \(D\).
    This places a lower bound on the sum of edge weights of 8 and the number of vertices are trivially lower bounded by 8.
    These necessary edges do not realise the distance of 2 between \(v_1, v_3\) and \(v_2, v_4\) and therefore additional edges are required.
    Furthermore, it is clear that these distances of 2 must be realised by paths of length 2 utilising none of these necessary edges (for otherwise certain required distances would be violated).
    Therefore, we need to attach an additional edge to each of the vertices \(v_1, v_2, v_3, v_4\) in order to realise the distance of 2 between \(v_1, v_3\) and \(v_2, v_4\).
    This requires a total additional sum of edge weights of 4 meaning we need at least 12 edge weight in total as shown in Figure \ref{fig:graphrealisationscomp2}.
    In the case of the combinatorial distance realisation problem, we require these 4 additional edges to be adjacent to at least 1 additional vertex, hence the minimum number of vertices is 9 as shown in Figure \ref{fig:graphrealisationscomp2} (right).
\end{proof}
\medskip

\noindent\textbf{Our results.}
We introduce notions and foundational results in Section \ref{sec:foundationalresults}.
Moreover, we discuss the straightforward polynomial time solution of \(0\)-\textsc{CombDMR}, due to Hakimi and Yau 1965 \cite{Hakimi1965DistanceMO}. 
In Section \ref{sec:kplus1} we provide a polynomial time algorithm to solve \(1\)-\textsc{CombDMR}, by solving an appropriate 2-SAT instance (see Algorithm \ref{alg:kp1} and Theorem \ref{thm:kp1}).
We then apply a similar construction of two appropriate 2-SAT instances in Section \ref{sec:kplus2} to provide a polynomial time algorithm for \(2\)-\textsc{CombDMR} (see Algorithm \ref{alg:kp2} and Theorem \ref{thm:kp2}).
Our main result of the paper is presented in Section \ref{sec:kplus3} and states that \(k\)-\textsc{CombDMR} is NP-complete for all fixed \(k \geq 3\) (see Theorem \ref{thm:maink}).
This is achieved by a reduction from the \(k\)-colourability problem.
Finally, in Section \ref{sec:tree}, on tree realisations of distance matrices, namely, the \textsc{TreeCombDMR} problem.
We show that the polynomial time algorithm of Culberson and Rudnicki 1989 \cite{Culberson1989} for weighted tree realisations of distance matrices can also produce solutions for \textsc{TreeCombDMR} through a simple modification (see Algorithm \ref{alg:cr1989Ours}).
Note that tree realisation problems are of particular interest in the field of Phylogenetics (see Semple and Steel 2003 \cite{Semple2003}).
Finally, we conclude and provide directions for future research in Section \ref{sec:conclusions}.

\section{Notions and Foundational Results} \label{sec:foundationalresults}

We begin by identifying the necessary and sufficient conditions for an input matrix \(D\) to admit at least one graph realisation.
\medskip
\begin{definition}[Distance matrix]
    \label{def:distancematrix}
    Let \(D\) be an \(n \times n\) matrix with non-negative integer valued entries.
    We call \(D\) a {\emph{distance matrix}} if it satisfies the following properties:
    \begin{itemize}
        \item[(i)] All diagonal entries of \(D\) are zero and all non-diagonal entries are strictly positive.
        \item[(ii)] \(D\) is a symmetric matrix.
        \item[(iii)] For all \(i,j,w \in [n]\), we have \[ D_{iw} + D_{wj} \ge D_{ij}. \]
    \end{itemize}
\end{definition}
\medskip

This definition gives rise to the following result.
\medskip
\begin{proposition}
    \label{prop:graphrealisationmetric}
    Let \(D\) be an \(n \times n\) matrix with non-negative integer valued entries. 
    \(D\) admits at least one graph realisation \((G, \Phi)\) if and only if \(D\) is a distance matrix.
\end{proposition}

\begin{proof}
    It is clear that (i), (ii) and (iii) need to be satisfied by any graph realisation \((G, \Phi)\) of \(D\).
    Conversely, if (i), (ii) and (iii) are satisfied by \(D\) then we can construct a graph realisation \((G, \Phi)\) of \(D\) as follows.
    Begin with with \(n\) isolated vertices \(v_1,\dots,v_n\).
    Let \(\Phi(i) = v_i\) for all \(i \in [n]\).
    The \(D\)-realising graph $G$ is then constructed by connecting any pair $v_i,v_j$, $1 \le i < j \le n$ by a path of length $D_{ij}$ of new interior vertices, such that, no two such paths have common interior vertices (with \(D_{ij}=1\) meaning that the vertices are adjacent).
    We call such paths {\emph{elementary paths}} of the graph \(G\).
    % By construction, we have that, in \(G\), \(d(v_i, v_j) \leq D_{ij}\) for all \(i < j\).
    Since any path \(\pi\) from \(v_i\) to \(v_j\) in \(G\) must be a concatenation of elementary paths, condition (iii) garantees that this concatenation is of length greater than or equal to \(D_{ij}\).
    % connecting \(v_i\) with \(v_{k_1}\), \(v_{k_1}\) with \(v_{k_2}\), \(\dots\) , \(v_{k_s}\) with \(v_{k_j}\), where \(s \in \N \cup \{0\}\) and \(k_1,\dots,k_s \in [n]\), its length \(l(\pi)\) must be \[l(\pi) = D_{i, k_1} + D_{k_1, k_2} + \cdots + D_{k_s,j} \geq D_{i,j},\]
    % because of condition (iii).
    This shows that \(d(v_i, v_j) \geq D_{ij}\).
    Furthermore, we know that there exists an elementary path \(\pi\) from \(v_i\) to \(v_j\) in \(G\) of length \(l(\pi) = D_{ij}\).
    Therefore, \(d(v_i, v_j) = D_{ij}\) for all \(i,j \in [n]\) and so \((G, \Phi)\) is a graph realisation of \(D\).
\end{proof}
\medskip

As a consequence of Proposition \ref{prop:graphrealisationmetric}, we will assume that \(D\) is a distance matrix with integer valued entries for all instances of \(k\)-\textsc{CombDMR}.
As the proof of Proposition \ref{prop:graphrealisationmetric} shows, we can always find a graph realisation of a distance matrix \(D\) with some number of additional vertices.
Another immediate consequence of the above construction is the following upper bound on the number of vertices for the existence of a graph realisation of \(D\).

\medskip
\begin{proposition}
    \label{lem:drealisinggraphbasic}
    Let \(D\) be an $n \times n$ distance matrix. 
    Then there exists a graph realisation $(G=(V,E),\Phi)$ of $D$ with 
    \[ |V| \leq n + \sum_{1 \le i < j \le n} (D_{ij}-1).\]
\end{proposition}
\medskip

We now seek to improve the result of Proposition \ref{lem:drealisinggraphbasic}, and in doing so, we introduce the following weighted graph.
\medskip
\begin{definition}[\(q\)-skeleton]
    \label{def:Gk}
    Let \(D\) be an \(n \times n\) distance matrix and \(q \in \N\).
    The {\emph{q-skeleton of D}} is the weighted graph \(G^q=(V^q, E^q, w)\) with vertices \(V^q = [n]\) and edges 
    \[E^q = \{\{i,j\} \in [n] \times [n] \mid (i < j) \land (D_{ij} \leq q)\},\] 
    that is, \(G^q\) has an edge between \(i\) and \(j\) if and only if \(D_{ij} \leq q\).
    Additionally, let the edge weights \(w:E^q \to \N\) be given by \[w(i,j) = D_{ij}, \quad \{i,j\} \in E^q.\]
    These edge weights are understood to be the lengths of the corresponding edges.
    Let \(d_{G^q}:V \times V \to \N_0 \cup \{\infty\}\) be the associated distance function of \(G^q\), that is, \(d_{G^q}(i,j)\) is the length of the shortest path between \(i\) and \(j\) in \(G^q\) and equal to \(\infty\) if no such path exists.
    The \(n \times n\) matrix \(D^{(q)}\), given by \(D^{(q)}_{ij} = d_{G^q}(i,j)\), is called the {\emph{distance matrix of the q-skeleton of D}}.
\end{definition}
\medskip
Of particular importance is the following fact:
\[D^{(q)}_{ij} \begin{cases} = D_{ij} & \, \text{if } D_{ij} \leq q, \\ \geq D_{ij} & \text{if } D_{ij} > q. \end{cases}\]
Notice, when we have \(D^{(q)} = D\) for some \(q\), then we can replace each edge \(\{i,j\} \in E^q\) with an elementary path of length \(D_{ij}\) in \(G^q\) to obtain a graph realisation of \(D\). 
Furthermore, \(D^{(q)}\) can be computed in polynomial time by any weighted all-pairs shortest-paths (APSP) algorithm (e.g. the Floyd-Warshall algorithm \cite[pp. 570–576]{Cormen1990}).
In fact, we have the following ordering of the matrices \(D^{(q)}\).

\begin{lemma}
    \label{lem:skeletonordering}
    Let \(D\) be an \(n \times n\) distance matrix and \(m = \max\{ D_{ij} : 1 \leq i < j \leq n \}\).
    Let \(D^{(q)}\) be the distance matrix of the q-skeleton of \(D\).
    Then we have,
    \[D^{(1)}_{ij} \geq D^{(2)}_{ij} \geq \cdots \geq D^{(m)}_{ij} = D_{ij} \quad \text{for all } i,j \in [n].\]
\end{lemma}

\begin{proof}
    In the transition from \(G^q\) to \(G^{q+1}\) the only possible change is that new edges are added to \(G^q\), decreasing the distance between vertices.
    Therefore, \(D^{(q+1)}_{ij} \leq D^{(q)}_{ij}\) for all \(i,j \in [n]\).
    Moreover, by the construction in the proof of Proposition \ref{prop:graphrealisationmetric} we have \(D^{(m)}_{ij} = D_{ij}\) for all \(i,j \in [n]\).
\end{proof}

Now let \(q_0\) denote the smallest \(q \in \N\) such that \(D^{(q)} = D\). 
Then we have the following improvement on Proposition \ref{lem:drealisinggraphbasic}.
\medskip
\begin{proposition}
    \label{lem:drealisinggraphbasic2}
    Let \(D\) be an $n \times n$ distance matrix and \(q_0 \in \N\) be the smallest \(q \in \N\) such that \(D^{(q)} = D\), where \(D^{(q)}\) the distance matrix of the q-skeleton of \(D\).
    Then there exists a graph realisation $(G=(V,E),\Phi)$ of \(D\) with 
    \begin{equation}
        \label{eq:drpminlowerbound2}
        |V| \leq n + \sum_{(1 \le i < j \le n) \land (2 \leq D_{ij} \leq q_0)} (D_{ij}-1).
    \end{equation}
\end{proposition}

\begin{proof}
    We construct the q-skeleton \(G^{q_0}\) as in Definition \ref{def:Gk}.
    We then convert \(G^{q_0}\) to a simple undirected graph \(G\) by replacing each edge \(\{i,j\} \in E^{q_0}\) with an elementary path of length \(D_{ij}\) in \(G^{q_0}\) to obtain a graph realisation of \(D\).
    Each elementary path is a path of length \(D_{ij}\) with \(D_{ij}-1\) new vertices.
    In \(G^{q_0}\) we have edges \(\{i,j\} \in E^{q_0}\) if and only if \(D_{ij} \leq q_0\), and so we have precisely the number of vertices as in \eqref{eq:drpminlowerbound2} in this graph realisation of \(D\).
\end{proof}
\medskip
Using the q-skeleton of \(D\) we can now provide a lower bound on the number of vertices required for a graph realisation of \(D\).
To do so, we generalise the earlier notion of elementary paths as follows: 
for a graph realisation \((G,\Phi)\) of \(D\), an {\emph{elementary path}} is a path of length \(D_{ij}\) between vertices \(v_i\) and \(v_j\) with no interior vertices in \(\Phi([n])\).
Then we have the following result.
\medskip
\begin{proposition}
    \label{prop:drpminlowerboundpre}
    Let \(s \in \N\), \(D\) be an \(n \times n\) distance matrix, \(D^{(q)}\) be the distance matrix of the q-skeleton of \(D\) and \((G, \Phi)\) be a graph realisation of \(D\).
    If there are no elementary paths of length greater than \(s\) in \(G\), then \(D^{(s)} = D\).
\end{proposition}
\begin{proof}
    Let \((G=(V,E),\Phi)\) be a graph realisation without elementary paths of length greater than \(s\) and without loss of generality assume that \(\Phi(i) = v_i \in V\) for all \(i \in [n]\).
    Let \(i,j \in [n]\) and \(\pi\) be a shortest path from \(v_i\) to \(v_j\) in \(G\) of length \(D_{ij}\).
    Then \(\pi\) is a concatenation of elementary paths \(\pi_1, \pi_2, \dots, \pi_r\) (\(r \in \N\)) where \(\pi_t\) is a path from \(v_{k_t}\) to \(v_{k_{t+1}}\) and has length \(D_{k_t k_{t+1}} = D^{(s)}_{k_t k_{t+1}}\leq s\).
    This implies,
    \[
        D_{ij} = d_G(v_i, v_j) = \sum_{t=1}^{r} d_G(v_{k_t}, v_{k_{t+1}}) = \sum_{t=1}^{r} D_{k_t k_{t+1}} = \sum_{t=1}^{r} D^{(s)}_{k_t k_{t+1}} \geq D^{(s)}_{ij}.
    \]
    Since \(D^{(s)}_{ij} \geq D_{ij}\) by Lemma \ref{lem:skeletonordering}, we have \(D^{(s)} = D\).
\end{proof}
\medskip
    
\medskip
\begin{proposition}
    \label{prop:drpminlowerbound}
    Let \(D\) be an \(n \times n\) distance matrix and \(q_0 \in \N\) be the smallest \(q \in \N\) such that \(D^{(q)} = D\), with \(D^{(q)}\) the distance matrix of the q-skeleton of \(D\).
    Any graph realisation \((G=(V,E), \Phi)\) of \(D\) must satisfy \(|V| \geq n + (q_0 - 1)\).
\end{proposition}
\begin{proof}
    Clearly, if \(q_0 = 1\), then the statement is trivially true.
    Now assume that \(q_0 \geq 2\) and \(D^{(q_0)} = D\).
    By Proposition \ref{prop:drpminlowerboundpre}, we know that in any graph realisation \((G, \Phi)\) of \(D\) there must exist an elementary path of length at least \(q_0\) between some pair of vertices \(v_i, v_j\) for \(i,j \in [n]\).
    The interior vertices of this elementary path must be contained within \(V \setminus \Phi([n])\), therefore we have \(|V|  \geq n + (q_0 - 1)\). 
\end{proof}
\medskip

Combining Proposition \ref{lem:drealisinggraphbasic2} and Proposition \ref{prop:drpminlowerbound} we obtain the following Proposition.
\medskip
\begin{proposition}
    \label{cor:simpleresult}
    Let \(D\) be an \(n \times n\) distance matrix and \(D^{(q)}\) be the distance matrix of the q-skeleton of \(D\).
    Then \(D\) has a graph realisation \((G, \Phi)\) with \(|V| = n\) if and only if \(D^{(1)} = D\).
\end{proposition}

\begin{proof}
    Proposition \ref{lem:drealisinggraphbasic2} implies that if \(D^{(1)} = D\) then there exists a graph realisation \((G, \Phi)\) of \(D\) with \(|V| = n\).
    Furthermore, Proposition \ref{prop:drpminlowerbound} implies that, if \(D^{(1)} \neq D\) (hence \(D^{(q)} = D\) for some \(q \geq 2\)), then any graph realisation \((G=(V,E), \Phi)\) of \(D\) must have \(|V| \geq n + 1\).
\end{proof}
\medskip
Due to its general significance throughout this paper, we introduce for any distance matrix \(D\) the unweighted graph \(G_D\), which is simply the 1-skeleton \(G^1\) of \(D\) as in Definition \ref{def:Gk} without the edge weights.
Note that Proposition \ref{cor:simpleresult} is equivalent to the following Theorem by Hakimi and Yau 1965 \cite{Hakimi1965DistanceMO}.
\medskip
\begin{theorem}
    \label{thm:nequalk}
    For any \(n \times n\) distance matrix \(D\), the following statements are equivalent:
    \begin{itemize}
    \item[-] A graph realisation \((G=(V,E), \Phi)\) of \(D\) with \(|V| = n\) exists.
    \item[-] A graph realisation of \(D\) is \((G_D=(V_D,E_D), \Phi)\) with \(V_D=v_1,\dots,v_n\) and \(\Phi(i) = v_i\) for \(i \in [n]\). 
    \end{itemize}
\end{theorem}
\medskip
Hakimi and Yau 1965 \cite{Hakimi1965DistanceMO} also show that the graph realisation of a distance matrix \(D\) with \(|V|=n\) is unique up to isomorphism.
Theorem \ref{thm:nequalk} gives us a simple algorithm to solve the \(0\)-\textsc{CombDMR} as follows:

\begin{itemize}
    \item[1.] Construct the graph \(G_D\).
    \item[2.] Check if the distance function of \(G_D\) coincides with \(D\).
\end{itemize}

If \(D\) has a graph realisation with \(|V|=n\), then we call \(D\) a {\emph{self-realising distance matrix}}.
Note that the graph \(G_D\) is always a subgraph of any graph realisation of \(D\), as shown in the following proposition.
\medskip
\begin{proposition}
    \label{prop:g_d_subgraph}
    Let \(D\) be an \(n \times n\) distance matrix and \(G_D=(V_D,E_D)\) be the associated unweighted graph, with \(V_D = \{v_1,\dots,v_n\}\).
    Then any graph realisation \((G, \Phi)\) of \(D\) with \(\Phi(i) = v_i\) for \(i \in [n]\) has \(G_D\) as the induced subgraph on the vertices \(v_1,\dots,v_n\).
\end{proposition}

\begin{proof}
    Let \((G, \Phi)\) be a graph realisation of \(D\) with \(\Phi(i) = v_i\) for \(i \in [n]\).
    Let \(G'=(V',E')\) be the induced subgraph on the vertices \(v_1,\dots,v_n\).
    We seek to show that \(E'=E_D\).
    Let \(d_{G'}\) be the distance function of \(G'\).
    Suppose that there exists an edge \(\{v_i, v_j\} \in E'\) such that \(\{v_i, v_j\} \not \in E_D\).
    Then, by construction, \(D_{ij} > 1\) and \(d_{G'}(v_i, v_j) =~1\) which is a contradiction.
    Likewise, suppose that there exists an edge \(\{v_i, v_j\} \in E_D\) such that \(\{v_i, v_j\} \not \in E'\).
    Then, by construction, \(D_{ij} = 1\) and \(d_{G'}(v_i, v_j) > 1\) which is also a contradiction.
    Therefore, \(E'=E_D\) and \(G'=G_D\) is an induced subgraph of \(G\) on the vertices \(v_1,\dots,v_n\).
\end{proof}
\medskip

\section{Polynomial solution of $1$-\textsc{CombDMR} \label{sec:kplus1}}

In this section we consider the case where \(k=1\).
Knowing that \(D\) must be a distance matrix by Proposition \ref{prop:graphrealisationmetric}, we then check whether \(D\) is self-realising by the application of Theorem \ref{thm:nequalk} and its associated algorithm.
If this is not the case, we know that \(|V|\geq n+1\) for any graph realisation \((G=(V,E), \Phi)\) of \(D\).

Now we aim to develop a polynomial time algorithm to decide the existence of a graph realisation \((G, \Phi)\) of \(D\) with \(k=1\).
Our algorithm will also provide such a graph realisation if it exists.
The approach we take to solve this problem involves a particular 2-Satisfiability (2-SAT) instance \cite{papadimitriou1994computational}.
The 2-SAT problem is the special case of the boolean satisfiability problem, where the input formula is in Conjunctive Normal Form (2-CNF), namely, every clause consists of at most two literals. 
Given a 2-CNF formula $\phi$, the question is whether there exists a truth assignment of the variables that satisfies $\phi$.
Formally, a 2-SAT instance is expressed as a conjunction of a set \(\mathcal{C}\) of clauses, that is
\[
\phi = \bigwedge_{c_i \in \mathcal{C}} c_i = c_1 \land \ldots \land c_m
\]
for some finite \(m \in \N_0\), where each clause \( c_i \) of \(\mathcal{C}\) is a disjunction of two literals:
\[
c_i = (\ell_i^{(1)} \lor \ell_i^{(2)}).
\]
Here, each of \( \ell_i^{(1)}, \ell_i^{(2)} \) is a  literal, i.e., either a variable $x$ or its negation $\bar{x}$.
The 2-SAT problem is known to be polynomial time solvable by Krom 1967 \cite{Krom1967}.

Within any graph realisation \((G=(V,E),\Phi)\) with \(|V|=n+1\) and \(\Phi(i)=v_i\) we know that the induced subgraph on vertices \(v_1,\dots,v_n\) agrees with \(G_D\) by Proposition~\ref{prop:g_d_subgraph}.
Therefore, we look at the addition of a new vertex \(v_{n+1}\) to \(G_D\) and construct a 2-CNF formula $\phi_1$ to determine which vertices \(v_1,\dots,v_n\) can and which cannot be adjacent to \(v_{n+1}\) in \(G\).
We construct $\phi_1$ as the disjunction of the clauses in \(\mathcal{C}_1\), as follows:
\begin{itemize}
    \item[1.] Let \(G_D = (V_D,E_D)\) be the graph as described in Proposition \ref{prop:g_d_subgraph} with vertices \(V_D = \{v_1,\dots,v_n\}\) and \(\Phi(i) = v_i\) for \(i \in [n]\).
    \item[2.] Create an empty set of clauses \(\mathcal{C}_1 = \emptyset\) and let \(\{x_{i,n+1} : i \in [n]\}\) be the set of boolean variables representing the existence of an additional edge \(\{v_i, v_{n+1}\}\) to those already in \(G_D\).
    That is, \(x_{i,n+1}\) is true if and only if \(v_i\) is adjacent to \(v_{n+1}\) in \(G\).
    \item[3.] For all \(i,j \in [n]\) with \(D_{ij} > 2\), we know that the vertices \(v_i\) and \(v_j\) must not both be adjacent to \(v_{n+1}\) in any graph realisation of \(D\), as otherwise this would result in a distance of~2 between \(v_i\) and \(v_j\).
    This condition is equivalent to the following clause being satisfied:    
    \begin{equation}
        \label{eq:k1_2sat3a}
        (\bar{x}_{i,n+1} \lor \bar{x}_{j,n+1}).
    \end{equation}
    We therefore add the clauses \eqref{eq:k1_2sat3a} to \(\mathcal{C}_1\) for all \(i,j \in [n]\) with \(D_{ij} > 2\).
    \item[4.] For all \(i,j \in [n]\), with, \(D_{ij} = 2\) and \(d_{G_D}(v_i, v_j) > 2\), the following boolean expression must be satisfied:
    \begin{equation}
        \label{eq:k1_2sat4a}
        (x_{i,n+1} \land x_{j,n+1}),
    \end{equation}
    meaning that the distance between \(v_i\) and \(v_j\) must be 2 and realised by a path of length 2 via \(v_{n+1}\).
    The boolean expression \eqref{eq:k1_2sat4a} is equivalent to the following two clauses both being satisfied:
    \begin{eqnarray}
        \label{eq:k1_2sat4b}
        (x_{i,n+1} \lor {x}_{i,n+1}),\\
        \label{eq:k1_2sat4c}
        (x_{j,n+1} \lor {x}_{j,n+1}).
    \end{eqnarray}
    Therefore, we add the clauses \eqref{eq:k1_2sat4b} and \eqref{eq:k1_2sat4c} to \(\mathcal{C}_1\) for all \(i,j \in [n]\) with \(D_{ij} = 2\) and \(d_{G_D}(v_i, v_j) > 2\) as there is no other way to realise a distance of 2 between \(v_i\) and \(v_j\) in \(G\).    
\end{itemize}

This completes the construction of the 2-SAT instance \(\phi_1\) of the clauses \(\mathcal{C}_1\).
We will see (Lemma~\ref{lem:k1eqd2} below) that any satisfying assignment {\bf{X}} of \(\phi_1\) gives rise to a graph \(G_{\phi_1,{\bf{X}}}\), whose distance function on \(v_1, \dots, v_n\) agrees with the distance matrix \(D^{(2)}\) of the 2-skeleton of \(D\).

We now introduce the following definition.
A truth assignment of boolean variables \(\{x_{i,j}\}\) is said to be {\emph{consistent}} with a graph \(G=(V,E)\) with \(V \supset \{v_1, \dots, v_m\}\) if
\begin{eqnarray*}
    x_{i,j} = \begin{cases}
        {\rm{True}} \quad {\rm{if}} \, \{v_i, v_j\} \in E,\\
        {\rm{False}} \quad {\rm{if} }\, \{v_i, v_j\} \notin E.
    \end{cases}
\end{eqnarray*}
\medskip
Since the clauses of the form \eqref{eq:k1_2sat3a}, \eqref{eq:k1_2sat4b} and \eqref{eq:k1_2sat4c} are all necessary conditions for a graph realisation of \(D\) with \(|V| \leq n+1\), we have the following observation.
\medskip
\begin{observation}
    \label{obs:prop:k1cneccessary}
    Let \(D\) be an \(n \times n\) distance matrix.
    If \(\phi_1\) is not satisfiable then no graph realisation \((G=(V,E), \Phi)\) of \(D\) with \(|V| \leq n+1\) exists.
\end{observation}
\medskip
Note that any graph realisation \((G=(V,E), \Phi)\) of \(D\) with \(|V| = n+1\) gives rise to a satisfying assignment \({\bf{X}}\) of \(\phi_1\) which is consistent with \(G\).
We now seek to show that we require only a single satisfying assignment of \(\phi_1\) to determine whether a graph realisation with \(|V|=n+1\) exists.
Let \(G_{\phi_1,{\bf{X}}}\) be the unique graph with vertex set \(V=\{v_1,\dots,v_{n+1}\}\) having \(G_D\) as the induced subgraph on the vertices \(v_1, \dots, v_n\) and consistent with a satisfying assignment \({\bf{X}}\) of \(\phi_1\).
For such a graph, let \(D(\phi_1,{\bf{X}})\) denote the \(n \times n\) distance matrix of \(G_{\phi_1,{\bf{X}}}\) over the vertices \(\{v_1,\dots,v_n\}\). 
Note that \(D(\phi_1,{\bf{X}})\) can be computed in polynomial time using any known all-pairs shortest-paths (APSP) algorithm. 

\medskip
\begin{lemma}
    \label{lem:k1eqd2}
    Let \({\bf{X}}\) be a satisfying assignment of $\{x_{i,n+1} : i \in [n]\}$ of \(\phi_1\) and \(D^{(2)}\) be the distance matrix of the 2-skeleton of \(D\).
    Then \[D(\phi_1,{\bf{X}}) = D^{(2)}.\]
\end{lemma}

\begin{proof}
    Let \({\bf{X}}\) be a satisfying assignment of \(\phi_1\).
    By construction, we know that \(G_D\) is an induced subgraph of \(G_{\phi_1,{\bf{X}}}\) and by clauses of the form \eqref{eq:k1_2sat3a}, \eqref{eq:k1_2sat4b} and \eqref{eq:k1_2sat4c} we have that,
    \begin{equation}
    \label{eq:lem:k1eqd2}
        D(\phi_1,{\bf{X}})_{ij} = D_{ij} = D^{(2)}_{ij} \quad \text{for all } i,j \in [n] \text{ with } D_{ij} \leq 2.
    \end{equation}

    Assume, for the sake of contradiction, that \(D(\phi_1,{\bf{X}})_{ij} \neq D^{(2)}_{ij}\) for some \(i,j \in [n]\) with \(D_{ij} > 2\).
    First assume that \(D(\phi_1,{\bf{X}})_{ij} < D^{(2)}_{ij}\).
    Let \(\pi\) be a shortest path between \(v_i\) and \(v_j\) in \(G_{\phi_1,{\bf{X}}}\) of length \(D(\phi_1,{\bf{X}})_{ij}\).
    Then \(\pi\) must be a concatenation of elementary paths \(\pi_1,\dots,\pi_s\) of lengths \(1\) or \(2\).
    Each elementary path of length \(1\) and \(2\) has between its endpoints a corresponding edge in the \(2\)-skeleton \(G^2\) of \(D\) of equivalent length.
    Therefore, in \(G^2\) there must exist a path of length \(D(\phi_1,{\bf{X}})_{ij}\), which is a contradiction to \(D(\phi_1,{\bf{X}})_{ij} < D^{(2)}_{ij}\).
    
    Now assume that \(D(\phi_1,{\bf{X}})_{ij} > D^{(2)}_{ij}\).
    Let \(\pi\) be a shortest path between vertices \(i \in [n]\) and \(j \in [n]\) in \(G^2\), the \(2\)-skeleton of \(D\) of length \(D^{(2)}_{ij}\).
    Then \(\pi\) must consist of consecutive edges \(e_1=\{w_1,w_2\},\dots,e_s=\{w_{s},w_{s+1}\}\) such that \(w_1 = i\) and \(w_{s+1} = j\), that is, \(D^{(2)}_{w_{t}w_{t+1}} \leq 2\) for all \(t \in [s]\).
    Then, as \(D(\phi_1,{\bf{X}})_{w_tw_{t+1}} = D^{(2)}_{w_tw_{t+1}}\) by \eqref{eq:lem:k1eqd2}, we know that for each edge \(e_t\) there must exist a corresponding path of length \(1\) or \(2\) between \(w_t\) and \(w_{t+1}\) in \(G_{\phi_1,{\bf{X}}}\).
    Therefore, we have a path of length \(D^{(2)}_{ij}\) between \(v_i\) and \(v_j\) in \(G_{\phi_1,{\bf{X}}}\) which is a contradiction to the assumption that \(D(\phi_1,{\bf{X}})_{ij} > D^{(2)}_{ij}\).
\end{proof}
\medskip
An immediate consequence of Lemma \ref{lem:k1eqd2} is the following corollary.
\medskip
\begin{corollary}%[assignment invariance]
    \label{cor:lem:k1dleqdprime}
    Let \({\bf{X}}\) and \({\bf{X}}'\) be two distinct satisfying assignments of \(\phi_1\). Then \[D(\phi_1,{\bf{X}}) = D(\phi_1,{\bf{X}}').\]
\end{corollary}
\medskip
Corollary \ref{cor:lem:k1dleqdprime} tells us that the distance matrix \(D(\phi_1,{\bf{X}})\)  of \(G_{\phi_1,{\bf{X}}}\) is invariant over all satisfying assignments of \(\phi_1\).
Therefore, if we can find a single satisfying assignment \({\bf{X}}\) of \(\phi_1\), then we can construct \(G_{\phi_1,{\bf{X}}}\), and if \(D(\phi_1,{\bf{X}}) = D\) then we have found a graph realisation \((G=(V,E), \Phi)\) of \(D\) with \(|V| = n+1\).
It remains to show that, in the case \(D(\phi_1,{\bf{X}}) \neq D\), no graph realisation of \(D\) with \(|V| = n+1\) exists.
\medskip
\begin{proposition}
    \label{prop:k1_2sat2_not_realised}
    Let \(D\) be an \(n \times n\) distance matrix.
    If \(D(\phi_1,{\bf{X}}) \neq D\) for some satisfying assignment $\bf{X}$ of \(\phi_1\) then no graph realisation \((G=(V,E), \Phi)\) of \(D\) with \(|V| = n+1\) exists.
\end{proposition}

\begin{proof}
    Assume that a graph realisation \((G=(V,E), \Phi)\) of \(D\) with \(|V| = n+1\) exists.
    By Observation \ref{obs:prop:k1cneccessary}, we know that there exists a satisfying assignment \({\bf{X}}'\) of \(\phi_1\), consistent with the graph \(G\), such that \(D(\phi_1,{\bf{X}}') = D\).
    By Corollary \ref{cor:lem:k1dleqdprime}, we know that \(D(\phi_1,{\bf{X}}) = D(\phi_1,{\bf{X}}')\) which is a contradiction to \(D(\phi_1,{\bf{X}}) \neq D\).
\end{proof}
\medskip
In summary, we have the following polynomial time algorithm to determine whether there exists a graph realisation \((G=(V,E), \Phi)\) of \(D\) with \(|V| = n+1\). Moreover, the algorithm produces such a graph realisation if it exists.
\medskip
\begin{alg}[Solving 1-\textsc{CombDMR}]{\phantom{Bob}}

    \label{alg:kp1}
    \noindent {\bf{Input:}} An \(n \times n\) distance matrix \(D\).\\
    {\bf{Output:}} A graph realisation \((G=(V,E), \Phi)\) of \(D\) with \(|V| = n+1\) or a statement that no such graph realisation exists.
    \begin{itemize}
        \item[1.] Let \(\Phi(i) = v_i\) for \(i \in [n]\). 
        \item[2.] Construct the 2-CNF formula $\phi_1 = \bigwedge_{c\in\mathcal{C}_1}c$ as described above.
        \item[3.] Compute a satisfying assignment \({\bf{X}}\) of $\phi_1$, if it exists.
        \item[4.] If \(\phi_1\) is not satisfiable then no graph realisation \((G=(V,E), \Phi)\) of \(D\) with \(|V| = n+1\) exists. (By Observation \ref{obs:prop:k1cneccessary})
        \item[5.] If \(\phi_1\) is satisfiable then construct the graph \(G_{\phi_1,{\bf{X}}}\) consistent with the satisfying assignment \({\bf{X}}\).
        \item[6.] Compute the distance matrix \(D(\phi_1,{\bf{X}})\) of \(G_{\phi_1,{\bf{X}}}\) (using any APSP algorithm).
        \item[7.] If \(D(\phi_1,{\bf{X}}) = D\) then \((G_{\phi_1,{\bf{X}}}, \Phi)\) is a graph realisation of \(D\) with \(|V| = n+1\).
        \item[] Otherwise, if \(D(\phi_1,{\bf{X}}) \neq D\) then no graph realisation of \(D\) with \(|V| = n+1\) exists. (By Proposition \ref{prop:k1_2sat2_not_realised})
    \end{itemize}
\end{alg}
\medskip
We then have the following Theorem.
\medskip
\begin{theorem}
    \label{thm:kp1}
    1-\textsc{CombDMR} is solvable in polynomial time. 
    Moreover, if the input distance matrix \(D\) is a \textsc{YES}-instance of 1-\textsc{CombDMR}, then also a graph realisation \((G=(V,E), \Phi)\) of \(D\) can be computed with the same running time.
\end{theorem}

\begin{proof}
    We proceed by proving the correctness of Algorithm \ref{alg:kp1} and its polynomial time complexity.
    Correctness follows from Proposition \ref{prop:k1_2sat2_not_realised}, Lemma \ref{lem:k1eqd2} and Observation~\ref{obs:prop:k1cneccessary}.
    The 2-CNF formula \(\phi_1\) can be constructed and solved in polynomial time.
    The distance matrix \(D(\phi_1,{\bf{X}})\) can be computed in polynomial time.
    Furthermore, Corollary \ref{cor:lem:k1dleqdprime} implies that we only require a single satisfying assignment of \(\phi_1\) to determine whether a graph realisation \((G=(V,E), \Phi)\) of \(D\) with \(|V| = n+1\) exists or not.
\end{proof}
\medskip

\section{Polynomial solution of $2$-\textsc{CombDMR} \label{sec:kplus2}}
In this section we consider $k$-\textsc{CombDMR} in the case where \(k=2\).
Knowing that \(D\) must be a distance matrix by Proposition \ref{prop:graphrealisationmetric}, we can check if \(k \in \{0,1\}\) is sufficient by application of Theorem \ref{thm:nequalk} and Theorem \ref{thm:kp1}, respectively.
If this is not the case, we must have \(|V|\geq n+2\) for any graph realisation \((G=(V,E), \Phi)\) of \(D\).

Now we aim to develop a polynomial time algorithm to decide the existence of a graph realisation \((G=(V,E),\Phi)\) of \(D\) with \(k = 2\). 
Additionally, our algorithm will also provide such a graph realisation if it exists.
The approach we take to this problem involves solving two particular 2-Satisfiability (2-SAT) instances.
In any graph realisation of \(D\) with \(|V| = n+2\) and \(\Phi(i) = v_i\) for \(i \in [n]\), the two additional vertices \(v_{n+1}\) and \(v_{n+2}\) are either not adjacent (\(\{v_{n+1}, v_{n+2}\} \notin E\)) or adjacent (\(\{v_{n+1}, v_{n+2}\} \in~E\)).
We consider these two cases separately as different 2-SAT instances, denoted by \(\phi_2\) and \(\phi_2'\), respectively.

First, let us consider the case where \(\{v_{n+1}, v_{n+2}\} \notin E\).
The associated 2-SAT instance \(\phi_2\) as a disjunction of clauses \(\mathcal{C}_2\), is constructed as follows:
\begin{itemize}
    \item[1.] Let \(G_D = (V_D,E_D)\) be the graph as described previously with vertices \(V_D = \{v_1,\dots,v_n\}\) and \(\Phi(i) = v_i\) for \(i \in [n]\).
    \item[2.] Create an empty list of clauses \(\mathcal{C}_2=\emptyset\) and let \(\{x_{i,j} : i \in [n], j \in \{n+1,n+2\}\}\) be the set of boolean variables representing the existence of an additional edge \(\{v_i, v_j\}\) to those already in \(G_D\) to the additional vertices \(v_{n+1}\) and \(v_{n+2}\).
    That is, \(x_{i,j}\) is true if and only if \(v_i\) is adjacent to \(v_j\).
    \item[3.] For all \(i,j \in [n]\) such that \(D_{ij} > 2\) we know that the vertices \(v_i\) and \(v_j\) must not both be adjacent to \(v_{n+1}\) or \(v_{n+2}\) in any graph realisation of \(D\) as otherwise this would result in a distance of 2 between \(v_i\) and \(v_j\).
    This condition is equivalent to the following clauses both being satisfied:    
    \begin{eqnarray}
        \label{eq:k2_2sat3a}
        (\bar{x}_{i,n+1} \lor \bar{x}_{j,n+1}),\\
        \label{eq:k2_2sat3b}
        (\bar{x}_{i,n+2} \lor \bar{x}_{j,n+2}).
    \end{eqnarray} 
    We therefore add the clauses \eqref{eq:k2_2sat3a} and \eqref{eq:k2_2sat3b} to \(\mathcal{C}_2\) for all \(i,j \in [n]\) such that \(D_{ij} > 2\).
    \item[4.] For all \(i,j \in [n]\), such that, \(D_{ij} = 2\) and \(d_{G_D}(v_i, v_j) > 2\), we know the following boolean expression must be satisfied:
    \begin{equation}
        \label{eq:k2_2sat4a}
        (x_{i,n+1} \land x_{j,n+1}) \lor ({x}_{i,n+2} \land {x}_{j,n+2}),
    \end{equation}
    meaning that the distance between \(v_i\) and \(v_j\) must be 2 and realised via a path of length 2 via \(v_{n+1}\) or \(v_{n+2}\).
    The boolean expression \eqref{eq:k2_2sat4a}, by distributivity, is equivalent to the following 4 clauses being satisfied:
    \begin{eqnarray}
        \label{eq:k2_2sat4b}
        (x_{i,n+1} \lor {x}_{i,n+2}),\\
        \label{eq:k2_2sat4c}
        (x_{j,n+1} \lor {x}_{j,n+2}),\\
        \label{eq:k2_2sat4d}
        (x_{i,n+1} \lor x_{j,n+2}),\\
        \label{eq:k2_2sat4e}
        (x_{i,n+2} \lor x_{j,n+2}).
    \end{eqnarray}
    Therefore, we add the clauses \eqref{eq:k2_2sat4b}, \eqref{eq:k2_2sat4c}, \eqref{eq:k2_2sat4d} and \eqref{eq:k2_2sat4e} to \(\mathcal{C}_2\) for all \(i,j \in [n]\) such that \(D_{ij} = 2\) and \(d_{G_D}(v_i, v_j) > 2\).    
\end{itemize}

This completes the construction of the 2-SAT instance \(\phi_2\) of clauses \(\mathcal{C}_2\).
As in the case of \(k=1\), we have the following analogous observation.
\medskip
\begin{observation}
    \label{obs:prop:k2cneccessary}
    Let \(D\) be an \(n \times n\) distance matrix.
    If the formula \(\phi_2\) is not satisfiable then no graph realisation \((G=(V,E), \Phi)\) of \(D\) with \(|V| \leq n + 2\) exists.
\end{observation}
\medskip
Note, in particular, that Observation \ref{obs:prop:k2cneccessary} is valid irrespective of whether there is an edge between \(v_{n+1}\) and \(v_{n+2}\).

As before, we will show that we require only a single satisfying assignment of \(\phi_2\) to determine if a graph realisation \((G=(V,E),\Phi)\) of \(D\) with \(|V| = n+2\) and \(\{v_{n+1}, v_{n+2}\} \notin E\) exists or not.
Therefore, we introduce \(G_{\phi_2,{\bf{X}}}\) to be the unique graph with vertex set \(V=\{v_1,\dots,v_{n+2}\}\) having \(\{v_{n+1}, v_{n+2}\} \notin E\) with \(G_D\) as the induced subgraph on the vertices \(v_1,\dots,v_{n}\)  and being consistent with a satisfying assignment~\({\bf{X}}\) of \(\phi_2\).
For such a graph, let \(D(\phi_2,{\bf{X}})\) denote the \(n \times n\) distance matrix of \(G_{\phi_2,{\bf{X}}}\) over the vertices \(\{v_1,\dots,v_n\}\).

\medskip
\begin{lemma}
    \label{lem:k2eqd2}
    Let \({\bf{X}}\) be an assignment of \(\{x_{i,j} : i \in [n], j \in \{n+1,n+2\}\}\)  that satisfies \(\phi_2\) and let \(D^{(2)}\) be the distance matrix of the 2-skeleton of \(D\).
    Then \[D(\phi_2,{\bf{X}}) = D^{(2)}.\]
\end{lemma}

\begin{proof}
    Let \({\bf{X}}\) be a satisfying assignment of \(\phi_2\).
    By construction, we know that \(G_D\) is an induced subgraph of \(G_{\phi_2,{\bf{X}}}\) and by clauses of the form \eqref{eq:k2_2sat3a}, \eqref{eq:k2_2sat3b}, \eqref{eq:k2_2sat4b}, \eqref{eq:k2_2sat4c}, \eqref{eq:k2_2sat4d} and \eqref{eq:k2_2sat4e} we have that,
    \begin{equation*}
        D(\phi_2,{\bf{X}})_{ij} = D_{ij} = D^{(2)}_{ij} \quad \text{for all } i,j \in [n] \text{ with } D_{ij} \leq 2.
    \end{equation*}
    
    Because of this property, the remainder of the proof is analogous to that of Lemma~\ref{lem:k1eqd2}.
\end{proof}
\medskip
As an immediate consequence of Lemma \ref{lem:k2eqd2}, we have the following Corollary.
\medskip
\begin{corollary}
    \label{cor:lem:k2dleqdprime}
    Let \({\bf{X}}\) and \({\bf{X}}'\) be two distinct satisfying assignments of \(\phi_2\) then \[D(\phi_2,{\bf{X}}) = D(\phi_2,{\bf{X}}').\]
\end{corollary}
\medskip
As in the case of \(1\)-\textsc{CombDMR}, we have the following Proposition.
\medskip
\begin{proposition}
    \label{prop:k2_2sat2_not_realised}
    Let \(D\) be an \(n \times n\) distance matrix.
    If \(D(\phi_2,{\bf{X}}) \neq D\) for at least one satisfying assignment \({\bf{X}}\) of \(\phi_2\) then no graph realisation \((G=~(V,E), \Phi)\) of \(D\) exists with \(|V| = n+2\) and \(\{v_{n+1}, v_{n+2}\} \notin E\), where \(v_{n+1}, v_{n+2}  \in V \setminus \Phi([n])\).
\end{proposition}

\begin{proof}
    The proof is analogous to that of Proposition \ref{prop:k1_2sat2_not_realised} and follows from Observation \ref{obs:prop:k2cneccessary} and Corollary \ref{cor:lem:k2dleqdprime}.
\end{proof}
\medskip

Therefore, if \(\phi_2\) has at least one satisfying assignment \({\bf{X}}\), then the condition \(D(\phi_2,{\bf{X}}) = D\) is sufficient to determine if a graph realisation \((G=(V,E), \Phi)\) of \(D\) with \(|V| = n+2\) and \(\{v_{n+1}, v_{n+2}\} \notin E\) exists.

\medskip
Now let us consider the case where \(\{v_{n+1}, v_{n+2}\} \in E\) and construct the corresponding 2-SAT instance \(\phi_2'\) with clauses \(\mathcal{C}_2'\).
We know by Observation \ref{obs:prop:k2cneccessary} that all clauses included in \(\mathcal{C}_2\) are necessary conditions for a graph realisation \((G=(V,E), \Phi)\) of \(D\) with \(|V| = n+2\) irrespective of \(v_{n+1}\) and \(v_{n+2}\) being adjacent or not.
Therefore, it is clear that all clauses of \(\mathcal{C}_2\) should be included in \(\mathcal{C}_2'\).
Furthermore, we have the following lemma.
\medskip
\begin{lemma}
    \label{lem:k2shortest}
    Let \(D\) be an \(n \times n\) distance matrix and \(i,j \in [n]\).
    Assume \(D_{ij} = 3\) and \(D_{ij}^{(2)} > 3\), where \(D^{(2)}\) is the distance matrix of the 2-skeleton of \(D\).
    In any graph realisation \((G=(V,E), \Phi)\) of \(D\) with \(V = \{v_i = \Phi(i): i \in [n]\} \cup \{v_{n+1},v_{n+2}\}\), any shortest path from \(v_i\) to \(v_j\) in \(G\) must be of the following form: \[v_i \to v_{n+1} \to v_{n+2} \to v_j \quad {\rm{or}}\quad v_i \to v_{n+2} \to v_{n+1} \to v_j.\]
\end{lemma}
\begin{proof}
    Let \((G=(V,E), \Phi)\) be a graph realisation of \(D\) with \(V = \{v_i = \Phi(i): i \in [n]\} \cup \{v_{n+1},v_{n+2}\}\).
    Since \(D_{ij}=3\), there exists a shortest path of the form \(v_i \to v_{s} \to v_{t} \to v_j\) for some \(s,t \in [n+2]\).
    If \(\{s,t\} \neq \{n+1,n+2\}\) then this shortest is a concatenation of elementary paths of length \(1\) or \(2\) and therefore \(D_{ij} = D^{(2)}_{ij}\), which is a contradiction to the assumption that \(D_{ij}^{(2)} > 3\).
\end{proof}
Now we construct the 2-SAT instance \(\phi_2'\) with clauses \(\mathcal{C}_{2}'\), where \(\mathcal{C}_{2}'\) contains all the clauses of \(\mathcal{C}_{2}\) and further the clauses obtained by the following process:
\begin{itemize}
    \item[1.] Calculate \(D^{(2)}\) the distance matrix of the 2-skeleton of \(D\).
    \item[2.] For all \(i,j \in [n]\) with \(D_{ij} > 3\), we know that, if $v_i$ is adjacent to $v_{n+1}$ then $v_j$ cannot be adjacent to $v_{n+2}$, as otherwise this would result in a distance of \(3\) between \(v_i\) and \(v_j\). Similarly, if $v_i$ is adjacent to $v_{n+2}$ then $v_j$ cannot be adjacent to $v_{n+1}$.
    This condition is equivalent to the following clauses both being satisfied:
    \begin{eqnarray}
        \label{eq:2sat6a}
        (\bar{x}_{i,n+1} \lor \bar{x}_{j,n+2}),\\
        \label{eq:2sat6b}
        (\bar{x}_{i,n+2} \lor \bar{x}_{j,n+1}).
    \end{eqnarray} 
    We therefore add the clauses \eqref{eq:2sat6a} and \eqref{eq:2sat6b} to \(\mathcal{C}_{2}'\) for all \(i,j \in [n]\) with \(D_{ij} > 3\).
    \item[3.] For all \(i,j \in [n]\) with \(D_{ij} = 3\) and \(D_{ij}^{(2)} > 3\), the following boolean expression must be satisfied (by Lemma \ref{lem:k2shortest}):
    \begin{equation}
        \label{eq:2sat7a}
        (x_{i,n+1} \land x_{j,n+2}) \lor ({x}_{i,n+2} \land {x}_{j,n+1}),
    \end{equation}
    meaning that the distance between \(v_i\) and \(v_j\) must be 3 and it must be realised via a path of length 3 through \(v_{n+1}\) and \(v_{n+2}\).
    The boolean expression \eqref{eq:2sat7a}, by distributivity, is equivalent to the following 4 clauses being satisfied:
    \begin{eqnarray}
        \label{eq:2sat7b}
        (x_{i,n+1} \lor {x}_{i,n+2}),\\
        \label{eq:2sat7c}
        (x_{j,n+1} \lor {x}_{j,n+2}),\\
        \label{eq:2sat7d}
        (x_{i,n+1} \lor x_{j,n+1}),\\
        \label{eq:2sat7e}
        (x_{i,n+2} \lor x_{j,n+2}).
    \end{eqnarray}
    Therefore, we add the clauses \eqref{eq:2sat7b}, \eqref{eq:2sat7c}, \eqref{eq:2sat7d} and \eqref{eq:2sat7e} to \(\mathcal{C}_2'\) for all \(i,j \in [n]\) with \(D_{ij} = 3\) and \(D^{(2)}_{ij} > 3\).
\end{itemize}
This completes the construction of the 2-SAT instance \(\phi_2'\) with clauses \(\mathcal{C}_2'\).
As with \(\phi_2\), we have the following observation for \(\phi_2'\).
\medskip
\begin{observation}
    \label{obs:prop:k2pcneccessary}
    Let \(D\) be an \(n \times n\) distance matrix.
    If \(\phi'_2\) is not satisfiable then no graph realisation \((G=(V,E), \Phi)\) of \(D\) exists with \(|V| = n + 2\) and \(\{v_{n+1}, v_{n+2}\} \in E\), where
    \(v_{n+1}, v_{n+2}  \in V \setminus \Phi([n])\).
\end{observation}
\medskip

Let \({\bf{X}}\) be a satisfying assignment of \(\phi_2'\).
Then \(G'_{\phi_2',{\bf{X}}}\) denotes the graph \(G_{\phi_2,{\bf{X}}}\) with the additional edge \(\{v_{n+1},v_{n+2}\}\).
Furthermore, \(D'(\phi_2',{\bf{X}})\) denotes the \(n \times n\) distance matrix of \(G'_{\phi_2',{\bf{X}}}\) over the vertices \(\{v_1,\dots,v_n\}\).
We have the following lemma.
\medskip
\begin{lemma}
    \label{lem:k2peqd3}
    Let \({\bf{X}}\) be a satisfying assignment of \(\{x_{i,j} : i \in [n], j \in \{n+1,n+2\}\}\) that satisfies \(\phi'_2\) and let \(D^{(3)}\) be the distance matrix of the 3-skeleton of \(D\).
    Then \[D'(\phi'_2,{\bf{X}}) = D^{(3)}.\]
\end{lemma}

\begin{proof}
    Let \({\bf{X}}\) be a satisfying assignment of \(\phi'_2\).
    By construction, we know that \(G_D\) is an induced subgraph of \(G_{\phi'_2,{\bf{X}}}\) and since \(\phi_2 \subseteq \phi_2'\) and the clauses of the form \eqref{eq:2sat6a}, \eqref{eq:2sat6b}, \eqref{eq:2sat7b}, \eqref{eq:2sat7c}, \eqref{eq:2sat7d} and \eqref{eq:2sat7e} are satisfied, we have that,
    
    \begin{equation}
        \label{eq:lem:k2peqd3}
        D'(\phi'_2,{\bf{X}})_{ij} = D_{ij} = D^{(3)}_{ij} \quad \text{for all } i,j \in [n] \text{ with } D_{ij} \leq 3.
    \end{equation}

    Assume, for the sake of contradiction, that \(D'(\phi_2',{\bf{X}})_{ij} \neq D^{(3)}_{ij}\) for some \(i,j \in [n]\) with \(D_{ij} > 3\).
    First assume that \(D'(\phi_2',{\bf{X}})_{ij} < D^{(3)}_{ij}\).
    Let \(\pi\) be a shortest path between \(v_i\) and \(v_j\) in \(G'_{\phi'_2,{\bf{X}}}\) of length \(D'(\phi_2',{\bf{X}})_{ij}\).
    Then \(\pi\) must be a concatenation of elementary paths \(\pi_1,\dots,\pi_s\) of lengths \(1\), \(2\) or \(3\).
    Each elementary path of length \(1\), \(2\) and \(3\) has between its endpoints a corresponding edge in the \(3\)-skeleton \(G^3\) of \(D\) of equal length.
    Therefore, in \(G^3\) there must exist a path of length \(D'(\phi_2',{\bf{X}})_{ij}\), which is a contradiction to \(D'(\phi_2',{\bf{X}})_{ij} < D^{(3)}_{ij}\).
    
    Now assume, \(D'(\phi_2',{\bf{X}})_{ij} > D^{(3)}_{ij}\).
    Let \(\pi\) be a shortest path between vertices \(i \in [n]\) and \(j \in [n]\) in \(G^3\), the \(3\)-skeleton of \(D\) of length \(D^{(3)}_{ij}\).
    Then \(\pi\) must consist of consecutive edges \(e_1=\{w_1,w_2\},\dots,e_s=\{w_{s},w_{s+1}\}\) such that \(w_1 = i\) and \(w_{s+1} = j\), that is \(D^{(3)}_{w_{t}w_{t+1}} \leq 3\) for all \(t \in [s]\).
    Then, as \(D'(\phi_2',{\bf{X}})_{w_tw_{t+1}} = D^{(3)}_{w_tw_{t+1}}\) by \eqref{eq:lem:k2peqd3}, we know that for each edge \(e_t\) there must exist a corresponding path of length \(1\), \(2\) or \(3\) between \(w_t\) and \(w_{t+1}\) in \(G'_{\phi_2',{\bf{X}}}\).
    % Then for each edge \(e_t\) we know that there must exist a corresponding path of length \(1\), \(2\) or \(3\) between \(w_t\) and \(w_{t+1}\) in \(G'_{\phi'_2,{\bf{X}}}\) as \(D'(\phi_2',{\bf{X}})_{w_tw_{t+1}} = D^{(3)}_{w_tw_{t+1}}\).
    Therefore, we have a path of length \(D^{(3)}_{ij}\) between \(v_i\) and \(v_j\) in \(G'_{\phi'_2,{\bf{X}}}\) which is a contradiction to \(D'(\phi_2',{\bf{X}})_{ij} > D^{(3)}_{ij}\).
\end{proof}
\medskip
As an immediate consequence of Lemma \ref{lem:k2peqd3} we have the following Corollary.
\medskip
\begin{corollary}
    \label{cor:lem:dleqdprimecprime}
    Let \({\bf{X}}\) and \({\bf{X}}'\) be two distinct satisfying assignments of \(\phi_2'\). Then \[D'(\phi_2',{\bf{X}}) = D'(\phi_2',{\bf{X}}').\]
\end{corollary}
\medskip
\medskip
\begin{proposition}
\label{prop:k2p_2sat2_not_realised}
    Let \(D\) be an \(n \times n\) distance matrix.
    If \(D'(\phi_2',{\bf{X}}) \neq D\) for at least one satisfying assignment \({\bf{X}}\) of \(\phi'_2\) then no graph realisation~\((G=(V,E), \Phi)\) of \(D\) exists with \(|V| = n+2\) and \(\{v_{n+1}, v_{n+2}\} \in E\), where \(v_{n+1}, v_{n+2}  \in V \setminus \Phi([n])\).
\end{proposition}

\begin{proof}
    Assume that a graph realisation \((G=(V,E), \Phi)\) of \(D\) with \(|V| = n+2\) and \(\{v_{n+1}, v_{n+2}\} \in E\) exists.
    By Observation \ref{obs:prop:k2pcneccessary}, we know that there exists a satisfying assignment \({\bf{X}}'\) of \(\phi'_2\), consistent with the graph \(G\), such that \(D'(\phi'_2,{\bf{X}}') = D\).
    By Corollary \ref{cor:lem:dleqdprimecprime}, we know that \(D'(\phi_2',{\bf{X}}) = D'(\phi'_2,{\bf{X}}')\) which is a contradiction to \(D'(\phi_2',{\bf{X}}) \neq D\).
\end{proof}
\medskip

In summary, we have the following polynomial time algorithm to determine whether there exists a graph realisation \((G=(V,E), \Phi)\) of \(D\) with \(|V| = n+2\). Moreover, note that the algorithm produces such a graph realisation if it exists.
\medskip
\begin{alg}[Solving \(2\)-\textsc{CombDMR}]{\phantom{Bob}}\\
    {\bf{Input:}} An \(n \times n\) distance matrix \(D\).\\
    {\bf{Output:}} A graph realisation \((G, \Phi)\) of \(D\) with \(|V| = n+2\) or a statement that no such graph realisation exists.
    \label{alg:kp2}
    \begin{itemize}
        \item[1.] Let \(\Phi(i) = v_i\) for \(i \in [n]\).
        \item[2.] Construct the 2-CNF formula \(\phi_2 = \bigwedge_{c \in \mathcal{C}_2} c\) as described above.
        \item[3.] Compute a satisfying assignment \({\bf{X}}\) of \(\phi_2\), if it exists.
        \item[4.] If \(\phi_2\) is not satisfiable then no graph realisation \((G, \Phi)\) of \(D\) with \(|V| = n+2\) exists. (By Observation \ref{obs:prop:k2cneccessary})
        \item[5.] If \(\phi_2\) is satisfiable then construct the graph \(G_{\phi_2,{\bf{X}}}\) consistent with the satisfying assignment \({\bf{X}}\).
        \item[6.] Compute the distance matrix \(D(\phi_2,{\bf{X}})\) of \(G_{\phi_2,{\bf{X}}}\) (using any APSP algorithm).
        \item[7.] If \(D(\phi_2,{\bf{X}}) = D\) then \((G_{\phi_2,{\bf{X}}}, \Phi)\) is a graph realisation of \(D\) with \(|V| = n+2\).
        \item[8.] Construct the 2-CNF formula \(\phi_2' = \bigwedge_{c \in \mathcal{C}_2'} c\) as described above.
        \item[9.] Compute a satisfying assignment \({\bf{X}}'\) of \(\phi_2'\), if it exists.
        \item[10.] If \(\phi_2'\) is not satisfiable then no graph realisation \((G, \Phi)\) of \(D\) with \(|V| = n+2\) and \(\{v_{n+1}, v_{n+2}\} \in E\) exists. (By Observation \ref{obs:prop:k2pcneccessary})
        \item[11.] If \(\phi_2'\) is satisfiable then construct the graph \(G'_{\phi_2',{\bf{X}}'}\) consistent with the satisfying assignment \({\bf{X}}'\).
        \item[12.] Compute the distance matrix \(D'(\phi_2',{\bf{X}}')\) of \(G'_{\phi_2',{\bf{X}}'}\) (using any APSP algorithm).
        \item[13.] If \(D'(\phi_2',{\bf{X}}') = D\) then \((G'_{\phi_2',{\bf{X}}'}, \Phi)\) is a graph realisation of \(D\) with \(|V| = n+2\) and \(\{v_{n+1}, v_{n+2}\} \in E\).
        \item[] Otherwise, no graph realisation \((G, \Phi)\) of \(D\) with \(|V| = n+2\) exists. (By Proposition~\ref{prop:k2p_2sat2_not_realised})      
    \end{itemize}

\end{alg}

\medskip
\begin{theorem}
    \label{thm:kp2}
    \(2\)-\textsc{CombDMR} is solvable in polynomial time.
    Moreover, if the input distance matrix \(D\) is a YES-instance of \(2\)-\textsc{CombDMR} then also a graph realisation \((G=(V,E), \Phi)\) of \(D\) with \(|V| = n+2\) can be computed with the same running time.
\end{theorem}

\begin{proof}
    We proceed by proving the correctness of Algorithm \ref{alg:kp2} and its polynomial time complexity.
    Correctness follows from Observation~\ref{obs:prop:k2cneccessary}, Observation~\ref{obs:prop:k2pcneccessary} and Proposition \ref{prop:k2p_2sat2_not_realised}.
    The 2-CNF formulas \(\phi_2\) and \(\phi_2'\) can be constructed and solved in polynomial time.
    The distance matrices \(D(\phi_2,{\bf{X}})\) and \(D'(\phi_2',{\bf{X}'})\) can be computed in polynomial time.
    Furthermore, Corollary \ref{cor:lem:k2dleqdprime} and Corollary \ref{cor:lem:dleqdprimecprime} imply that we only require a single satisfying assignment of \(\phi_2\) and \(\phi_2'\), respectively, to determine whether a graph realisation \((G=(V,E), \Phi)\) of \(D\) with \(|V| = n+2\) exists or not.
\end{proof}
\medskip

\section{$k$-\textsc{CombDMR} is NP-complete for \(k \geq 3\)} \label{sec:kplus3}

In this section we prove that $k$-\textsc{CombDMR} is NP-complete for every \(k \geq 3\), via a reduction from \(k\)-colourability, which is known to be NP-complete~\cite{Karp1972,Lovasz1973,Stockmeyer1973}. 
For the readers' convenience we restate the \(k\)-colourability problem as follows:
\medskip
\begin{problem}[\(k\)-\textsc{colourability}]
    \label{prob:kcolour}
    Given a graph \(G = (V,E)\). 
    Does there exist a function \(\chi:V \to [k]\) such that for all \(\{i,j\} \in E\) we have \(\chi(i) \neq \chi(j)\)?
\end{problem}
\medskip

% Note \(k\)-colourability reduces trivially to the colourability problem, given as follows:
% \medskip
% \begin{problem}[\textsc{colourability}]
%     \label{prob:colour}
%     Given a graph \(G = (V,E)\) and \(k \in \N\).
%     Does there exist a function \(\chi:V \to \{1,\dots,k\}\) such that for all \(\{i,j\} \in E\) we have \(\chi(i) \neq \chi(j)\)?
% \end{problem}
% \medskip

% Karp 1972 \cite{Karp1972} showed that colourability is NP-complete.
% Lov\'{a}sz 1973 \cite{Lovasz1973} showed that colourability reduces to \(3\)-colourability.
% Therefore, \(3\)-colourability is NP-complete.
% Independently, Stockmeyer 1973 \cite{Stockmeyer1973} proved NP-completeness of \(3\)-colourability via a reduction from the \(3\)-\textsc{SAT} problem.
% Finally, note that \(k\)-colourability reduces to \((k+1)\)-colourability: \(G=(V,E)\) is \(k\)-colourable if and only if \(G'\) is \((k+1)\)-colourable, where \(G'\) is the suspension of \(G\).
% Therefore, by induction, \(k\)-colourability is NP-complete for all \(k \geq 3\).

As we will prove, \(k\)-\textsc{Colourability} can be reduced to the \(k\)-\textsc{CombDMR} problem by the following reduction algorithm.
\medskip
\begin{alg}[Reduction of $k$-\textsc{Colourability} to \(k\)-\textsc{CombDMR}]{\phantom{Bob}}

\label{alg:k3redution}
\noindent{\bf{Input:}} A connected simple undirected graph \(G_c = (V_c, E_c)\) for which we want to determine if it is \(k\)-colourable.\\
{\bf{Output:}} A distance matrix \(D\) such that \(G_c\) is \(k\)-colourable if and only if \(k\)-\textsc{CombDMR} for \(D\) is a YES-instance.
\begin{itemize}
    \item[1.] Enumerate the vertices of \(G_c\) such that \(V_c = \{v_1,\dots, v_{n_c}\}\) where \(n_c = |V_c|\).
    \item[2.] Construct the gadget graph \(G_g = (V_g, E_g)\), with \(V_c \subseteq V_g\), as follows. We subdivide each edge in $E_c$ twice, i.e., we replace each edge by a path of length 3 (containing two new vertices).
    For every pair of non-adjacent vertices in $G_c$, we add a path of length 2 between them (containing one new vertex). 
    We enumerate the vertices of \(G_g\) such that \(V_g = \{v_1,\dots, v_{n_c}, v_{n_c+1} \dots, v_{n_g}\}\) where $v_{n_c+1}, \dots, v_{n_g}$ are the new vertices and \(n_g = |V_g|\), see Figure~\ref{fig:k3reduction}.
    \item[3.] Let \(d_{G_g}\) denote the shortest path distance function of \(G_g\). 
    Construct the \(n \times n\) distance matrix \(D\) where \(n=n_g + 1\), with entries,
    \begin{align*}
    D_{ij} &= d_{G_g}(v_i, v_j) && \text{ for } i,j \in [n_g],\\
    D_{i'n} &= D_{ni'} = 2 && \text{ for } i' \in [n_c],\\
    D_{i''n} &= D_{ni''}=3 && \text{ for } i''\in[n_g] \setminus [n_c],\\
    D_{nn} &= 0.
    \end{align*}
    This will result in a distance matrix of the form:
    \begin{equation*}
        \label{eq:dk3}
        D =
        \left[
        \begin{array}{*{8}{c}|c}
        & & & & & & & & 2 \\
        & & & & & & & & \vdots \\
        & & & & \makebox[0pt]{\(D_{ij} = d_{G_g}(v_i, v_j)\)} & & & & 2 \\
        \cline{9-9}
        & & & & \makebox[0pt]{\(i,j \in [n_g]\)} & & & & 3 \\
        & & & & & & & & \vdots \\
        & & & & & & & & 3 \\
        \hline
        2 & \cdots & \cdots & 2 & \multicolumn{1}{|c}{3} & \cdots & \cdots & 3 & 0 \\
        \end{array}
        \right]
    \end{equation*}

\end{itemize}
\end{alg}
An example of the construction of a gadget graph as in Algorithm \ref{alg:k3redution} is illustrated in Figure \ref{fig:k3reduction}.

\begin{figure}[h!]
    \centering
    \begin{minipage}[t]{0.8\linewidth}
        \centering
        \begin{tikzpicture}
            % Nodes
            \node[draw, circle] (1) at (0,6) {\bf{$v_1$}};
            \node[draw, circle] (2) at (2,6) {\bf{$v_2$}};
            \node[draw, circle] (3) at (2,8) {\bf{$v_3$}};
            \node[draw, circle] (4) at (0,8) {\bf{$v_4$}};
            \node (10) at (-3,7) {Input graph \(G_c\)};

            % Original Nodes
            \node[draw, circle] (1b) at (0,0) {{\(v_1\)}};
            \node[draw, circle] (2b) at (4,0) {{\(v_2\)}};
            \node[draw, circle] (3b) at (4,4) {{\(v_3\)}};
            \node[draw, circle] (4b) at (0,4) {{\(v_4\)}};
            \node (10b) at (-3,2) {Gadget graph \(G_g\) of \(G_c\)};

            % add arrow from input to gadget which starts a bit above the input and ends a bit below the gadget
            \draw[->] (10) -- (10b);

            % Edges
            \draw (1) -- (2);
            \draw (2) -- (3);
            \draw (3) -- (4);
            \draw (4) -- (1);
            \draw (1) -- (3);

            % Additional Nodes on Edges
            \node[draw, circle] (a) at (1,0) {\(v_5\)}; % between 1 and 2
            \node[draw, circle] (b) at (3,0) {\(v_6\)}; % between 1 and 2
            \node[draw, circle] (c) at (4,1) {\(v_7\)}; % between 2 and 3
            \node[draw, circle] (d) at (4,3) {\(v_8\)}; % between 2 and 3
            \node[draw, circle] (e) at (3,4) {\(v_9\)}; % between 3 and 4
            \node[draw, circle] (f) at (1,4) {\(v_{10}\)}; % between 3 and 4
            \node[draw, circle] (g) at (0,1) {\(v_{11}\)}; % between 4 and 1
            \node[draw, circle] (h) at (0,3) {\(v_{12}\)}; % between 4 and 1
            \node[draw, circle] (i) at (1,1) {\(v_{13}\)}; % first between 1 and 3
            \node[draw, circle] (j) at (3,3) {\(v_{14}\)}; % second between 1 and 3
            \node[draw, circle] (l) at (1,3) {\(v_{15}\)}; % between 1 and 4

            % Edges
            \draw (1b) -- (a) -- (b) -- (2b);
            \draw (2b) -- (c) -- (d) -- (3b);
            \draw (3b) -- (e) -- (f) -- (4b);
            \draw (4b) -- (h) -- (g) -- (1b);
            \draw (1b) -- (i) -- (j) -- (3b);
            \draw (2b) -- (l) -- (4b);
        \end{tikzpicture}
    \end{minipage}
    \caption{Example of construction of a gadget graph \(G_g\) from an input graph \(G_c\) as in Algorithm \ref{alg:k3redution} with old vertices \(v_1,\dots,v_4\) (i.e., \(n_c = 4\)) and new vertices \(v_5,\dots,v_{15}\) (i.e., \(n_g = 15\)). \label{fig:k3reduction}}
\end{figure}
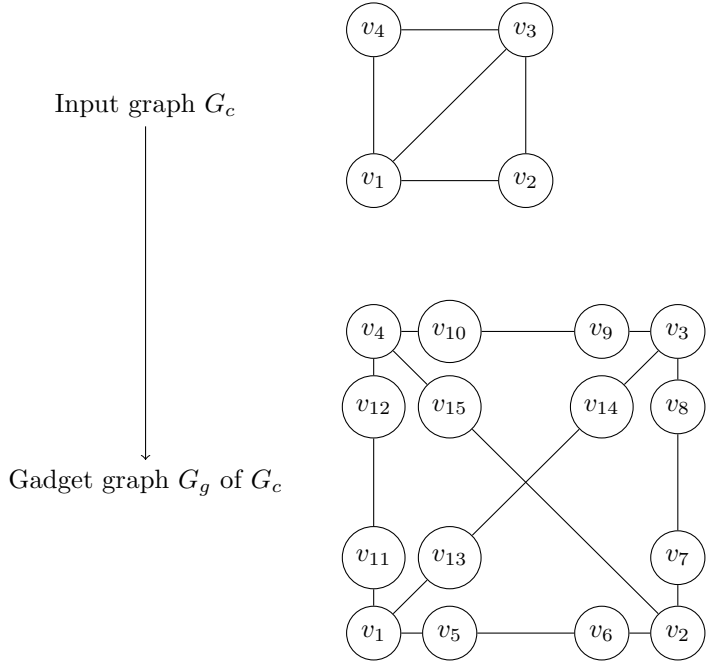
\FloatBarrier
We first prove that the constructed matrix \(D\) is always a distance matrix.

\medskip
\begin{proposition}
    The constructed matrix \(D\) of Algorithm \ref{alg:k3redution} satisfies the conditions of Definition \ref{def:distancematrix} and is therefore a distance matrix.
\end{proposition}

\begin{proof}
    It is clear that conditions (i) and (ii) of Definition \ref{def:distancematrix} are satisfied by construction.
    Therefore, it remains to show that the triangle inequality, condition (iii) of Definition \ref{def:distancematrix}, is satisfied.
    Within the submatrix \(D_{[n_g][n_g]}\) of the distance matrix \(D\), we know that the triangle inequality is satisfied, as \(G_g\) is a connected simple graph.
    Now we consider column \(n\) and row \(n\) of the matrix \(D\).
    For any \(i,j \in [n_g], i \neq j\), we know that 
    \begin{equation*}
        \label{eq:dist1}
        % Cases statement for distance function
        D_{in} + D_{nj} =
        \begin{cases}
            4 & \text{if } |\{i,j\} \cap [n_c]| = 2,\\
            5 & \text{if } |\{i,j\} \cap [n_c]| = 1,\\ 
            6 & \text{if } |\{i,j\} \cap [n_c]| = 0.
        \end{cases}
    \end{equation*}
    Consider now the distance function \(d_{G_g}\) of \(G_g\). 
    Let now \(i, j \in [n_g]\) and \(i \neq j\).
    If both $v_i, v_j$ are old vertices (i.e.~whenever $i,j\in [n_c]$) we have that
    \begin{equation*}
        \label{eq:dist2}
        % Cases statement for distance function
        d_{G_g}(v_i, v_j) =
        \begin{cases}
            2 & \text{if }  \{v_i, v_j\} \notin E_c, \\
            3 & \text{if }  \{v_i, v_j\} \in E_c.
        \end{cases}
    \end{equation*}
    If exactly one of $v_i, v_j$ is old and the other one is new (i.e.~$|\{i,j\} \cap [n_c]| = 1$), we know that \(d_{G_g}(v_i, v_j) \in \{1,2,3,4\}\). 
    Finally, if both vertices are new (i.e.~$|\{i,j\} \cap [n_c]| = 0$), we know that \(d_{G_g}(v_i, v_j) \in \{1,2,3,4,5\}\). 
    Hence, \(D_{in} + D_{nj} > d_{G_g}(v_i, v_j)\) for all \(i,j \in [n_g]\), and therefore the condition (iii) of Definition \ref{def:distancematrix} is satisfied.
\end{proof}
\medskip
Now that we have established that the constructed matrix \(D\) is a valid distance matrix, our next aim is to prove that, 
if \(D\) is a YES-instance of \(k\)-\textsc{CombDMR} then \(G_c\) is \(k\)-colourable (Proposition \ref{prop:k3redutionforward} below).
We start with the following useful lemma.
\medskip
\begin{lemma}
    \label{lem:notadjacentk3}
    Given an input graph \(G_c = (V_c, E_c)\) and \(k \in \N\).
    Let \((G=(V,E), \Phi)\) be any graph realisation of the constructed \(n \times n\) distance matrix \(D\) by Algorithm \ref{alg:k3redution}, with \(|V|=n+k\) vertices.
    Let \(v_{n+1},\dots,v_{n+k} \in V \setminus \Phi([n])\).
    Any two vertices \(v_i\) and \(v_j\) adjacent in the input graph \(G_c\) cannot both be adjacent to the same vertex \(v \in \{v_{n+1}\),\dots, \(v_{n+k}\}\) in \(G\). 
\end{lemma}

\begin{proof}
    Given a graph realisation \((G=(V,E), \Phi)\) of \(D\) with \(|V| = n+k\) vertices and \(\Phi(i)=v_i\) for \(i \in [n]\), with \(D\) constructed by Algorithm \ref{alg:k3redution} and let \(d_{G}\) denote the shortest path distance function of \(G\).
    As \(v_i\) and \(v_j\) are adjacent in \(G_c\), by construction \(D_{ij} = 3\) and \(d_{G}(v_i, v_j) = 3\).
    If \(v_i\) and \(v_j\) were both adjacent to the same vertex \(v \in \{v_{n+1}\),\dots, \(v_{n+k}\}\) in \(G\) then \(d_{G}(v_i, v_j) \leq 2\), which would be a contradiction to \(G\) being a graph realisation of \(D\).
\end{proof}
\medskip

\medskip
\begin{proposition}
    \label{prop:k3redutionforward}
    Given an input graph \(G_c = (V_c, E_c)\) and \(D\) the \(n \times n\) matrix as constructed in Algorithm \ref{alg:k3redution}.
    If \(D\) is a YES-instance of \(k\)-\textsc{CombDMR} then \(G_c\) is \(k\)-colourable.
\end{proposition}

\begin{proof}
    Assume a graph realisation \((G=(V,E), \Phi)\) of \(D\) as constructed by Algorithm \ref{alg:k3redution} with \(|V| = n+k\) exists.
    Without loss of generality let \(\Phi(i) = v_i\) for \(i \in [n]\) and let \(\{v_{n+1},\dots,v_{n+k}\} = V \setminus \Phi([n])\).
    Then, by Lemma \ref{lem:notadjacentk3} we know that any two adjacent vertices in \(G_c\) cannot both be adjacent to the same vertex \(v \in \{v_{n+1}, \dots, v_{n+k}\}\) in \(G\).
    Furthermore, we know that each vertex \(v_i\) for \(i \in [n_c]\) must be adjacent to at least one of the vertices \(v \in \{v_{n+1}, \dots, v_{n+k}\}\) in \(G\) to realise the \(D_{in} = D_{ni} = 2\) distances in \(D\).
    We construct a colouring of the vertices of \(G_c\) by assigning a colour to each of the vertices \(v_{n+1}, \dots, v_{n+k}\) and then assign the same colour to any vertex \(v_i\) for \(i \in [n_c]\) which is adjacent to that vertex (with arbitrary choice in the case of multiple adjacent vertices \(v_{n+1}, \dots, v_{n+k}\)).
    This is a valid \(k\)-colouring due to Lemma \ref{lem:notadjacentk3}.
\end{proof}
\medskip
The colour assignment in the above proof is illustrated in Figure \ref{fig:k3realisation} as a continuation of the example in Figure \ref{fig:k3reduction}.
\begin{figure}[h!]

    \begin{minipage}[t]{0.95\linewidth}
        \centering
        \begin{tikzpicture}
            % increase the distance between the nodes
            \tikzstyle{every node}=[node distance=1.5cm]
            % Original Nodes
            \node[draw, circle] (1) at (0,0) {\({v_1}\)};
            \node[draw, circle] (2) at (4,0) {\({v_2}\)};
            \node[draw, circle] (3) at (4,4) {\({v_3}\)};
            \node[draw, circle] (4) at (0,4) {\({v_4}\)};
            \node[draw, circle] (k) at (-6,2) {$v_{n}$};
            \node[draw, circle, fill=yellow] (k1) at (-4,4) {$v_{n+1}$};
            \node[draw, circle, fill=brown] (k2) at (-4,2) {$v_{n+2}$};
            \node[draw, circle, fill=pink] (k3) at (-4,0) {$v_{n+3}$};
        
            % Additional Nodes on Edges
            \node[draw, circle] (a) at (1,0) {\(v_{5}\)}; % between 1 and 2
            \node[draw, circle] (b) at (3,0) {\(v_{6}\)}; % between 1 and 2
            \node[draw, circle] (c) at (4,1) {\(v_{7}\)}; % between 2 and 3
            \node[draw, circle] (d) at (4,3) {\(v_{8}\)}; % between 2 and 3
            \node[draw, circle] (e) at (3,4) {\(v_{9}\)}; % between 3 and 4
            \node[draw, circle] (f) at (1,4) {\(v_{10}\)}; % between 3 and 4
            \node[draw, circle] (g) at (0,1) {\(v_{11}\)}; % between 4 and 1
            \node[draw, circle] (h) at (0,3) {\(v_{12}\)}; % between 4 and 1
            \node[draw, circle] (i) at (1,1) {\(v_{13}\)}; % first between 1 and 3
            \node[draw, circle] (j) at (3,3) {\(v_{14}\)}; % second between 1 and 3
            \node[draw, circle] (l) at (1,3) {\(v_{15}\)}; % between 1 and 4
    
            % Edges
            \draw (1) -- (a) -- (b) -- (2);
            \draw (2) -- (c) -- (d) -- (3);
            \draw (3) -- (e) -- (f) -- (4);
            \draw (4) -- (h) -- (g) -- (1);
            \draw (1) -- (i) -- (j) -- (3);
            \draw (2) -- (l) -- (4);
    
            % Additional Edges
            \draw (k) -- (k1) -- (4);
            \draw (k) -- (k1) -- (2);
            \draw (k) -- (k2) -- (1);
            \draw (k) -- (k3) -- (3);
    
        \end{tikzpicture}
        \caption{A graph realisation of \(D\) constructed from the example in Figure \ref{fig:k3reduction} with \(k=3\). 
        In accordance with the proof of Proposition \ref{prop:k3redutionforward}, the vertex \(v_{1}\) inherits the colour of vertex \(v_{n+2}\), the vertices \(v_{2}\) and \(v_{4}\) inherit the colour of vertex \(v_{n+2}\), and the vertex \(v_{3}\) inherits the colour of vertex \(v_{n+3}\).
        \label{fig:k3realisation}}
    \end{minipage}
    \end{figure}
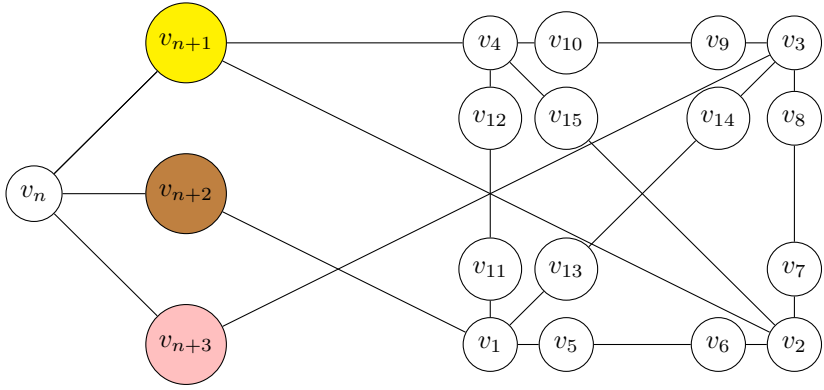

\medskip
The following theorem states that the implication in Proposition \ref{prop:k3redutionforward} is, in fact, an equivalence.
\medskip
\begin{theorem}
    \label{thm:3colmain}
    Let \(k \in \N\), \(G_c\) be an input graph for Algorithm \ref{alg:k3redution} and \(D\) be the constructed distance matrix.
    Then \(G_c\) is \(k\)-colourable if and only if \(D\) is a YES-instance of \(k\)-\textsc{CombDMR}.
\end{theorem}

\begin{proof}
    The forward direction is given by Proposition \ref{prop:k3redutionforward}.
    It remains to prove that if \(G_c\) is \(k\)-colourable then \(D\) is a YES-instance of \(k\)-\textsc{CombDMR}.
    Let \(G_c\) have a \(k\)-colouring \(\chi:V_c \to [k]\).
    We begin by constructing a graph realisation \((G=(V,E), \Phi)\) with \(V=\{v_1, \dots, v_{n+k}\}\) of the \(n \times n\) distance matrix \(D\).
    Let \(\Phi(i) = v_i\) for \(i \in [n]\).
    The edge set \(E\) of \(G\) is determined by the following requirements:
    \begin{itemize}
        \item The induced subgraph of \(G\) on the vertices \(\{v_1, \dots, v_{n_g}\}\) coincides with the gadget graph \(G_g\).
        \item \(v_n\) is not adjacent to any of the vertices of the gadget graph \(G_g\).
        \item For \(j \in [k]\) the neighbours of \(v_{n+j}\) are precisely the following: \(v_n\) and all vertices \(v_i\) in \(\{v_1, \dots, v_{n_c}\}\) whose colour is \(j\), that is, \(\chi(v_i) = j\).
    \end{itemize}

    We now show that \((G, \Phi)\) is indeed a graph realisation of \(D\).
    As each vertex \(v_i\) for \(i \in [n_g] \setminus [n_c]\) is adjacent to some \(v_j\) for \(j \in [n_c]\) in \(G\) and not adjacent to any vertex in \(\{v_{n}, \dots, v_{n+k}\}\), it suffices to verify the following equalities:
    \begin{align}
        d_G(v_n, v_i) &= 2 && \text{for } i \in [n_c],\label{eq:kcol1}\\
        d_G(v_i, v_j) &= D_{ij} && \text{for } i,j \in [n_c]. \label{eq:kcol2}
    \end{align}
    Let \(i \in [n_c]\). By construction, \(v_i\) is adjacent to \(v_{n+j}\) with \(j = \chi(v_i)\) and \(v_{n+j}\) is adjacent to \(v_n\), therefore \(d_G(v_n, v_i) \leq 2\) and \eqref{eq:kcol1} follows from the fact that \(v_n\) is not adjacent to \(v_i\).
    For \eqref{eq:kcol2}, we distinguish between two cases: if \(v_i\) and \(v_j\) are adjacent in \(G_c\) then \(D_{ij} = 3\) and \(d_{G_{g}}(v_i, v_j) = 3\) by construction.
    Moreover, \(\chi(v_i) \neq \chi(v_j)\) implies that \(v_i\) and \(v_j\) are not adjacent to the same vertex in \(\{v_{n+1}, \dots, v_{n+k}\}\) in \(G\).
    Therefore, there is no shortest path of length smaller than 3 between \(v_i\) and \(v_j\) in \(G\).
    If \(v_i\) and \(v_j\) are not adjacent in \(G_c\) then \(D_{ij} = 2\) and \(d_{G_{g}}(v_i, v_j) = 2\) by construction.
    We do not add an edge between \(v_i\) and \(v_j\) in \(G\), and therefore \(d_G(v_i, v_j) = 2 = D_{ij}\).
    Hence, \((G, \Phi)\) is a graph realisation of \(D\).
\end{proof}
\medskip

Note that \(k\)-\textsc{CombDMR} \(\in\) NP since any distance matrix of a finite graph can be computed in polynomial time. 
Therefore, Theorem \ref{thm:3colmain} implies the next theorem.

\medskip

\begin{theorem}
    \label{thm:maink}
    \(k\)-\textsc{CombDMR} is NP-complete for all \(k \in \N, k \geq 3\).
\end{theorem}

\section{Tree Realisations} \label{sec:tree}
In this section, we restrict our considerations to graph realisations which are trees.
For a given \(n \times n\) distance matrix \(D\), we call a graph realisation \((G, \Phi)\) of \(D\) a {\emph{tree realisation}} of \(D\) if \(G\) is a tree.
In contrast to Proposition \ref{lem:drealisinggraphbasic} for the general graph realisation problem, it is no longer true that any distance matrix admits always a tree realisation with sufficiently many vertices. 
To see this, observe that the distance matrix of Figure~\ref{fig:graphrealisationscomp1} can only be represented by a graph that contains a triangle, and thus is not a tree. 
Therefore we have the following problem.
\medskip
\begin{problem}
    \label{prob:drptree}
    \textsc{Tree Combinatorial Distance Matrix Realisation Problem \\(TreeCombDMR)}\\
    \emph{Input:} An $n \times n$ matrix $D$ with non-negative integer values. \\    
    \emph{Question:} Does there exist a simple (unweighted) tree $T=(V,E)$ with $|V| \geq n$ and an injective mapping $\Phi: [n] \rightarrow V$ such that the shortest-path distance function $d$ in $T$ satisfies
    \begin{equation*}
        \label{eq:drp-new-tree}
        d(\Phi(i), \Phi(j)) = D_{ij}
    \end{equation*}
        for all \(i,j \in [n]\)?
\end{problem}
\medskip
This problem is simpler than the general \textsc{$k$-CombDMR} problem.
In fact, this problem can be solved in \(O(n^4)\) time by a result by Zarecki\u{\i} 1965 \cite{Zareckii1965}.
Before we state this result, we first introduce the notion of a minimal tree realisation of a distance matrix.
\medskip
\begin{definition}[Minimal Tree Realisation]
    \label{def:minimaltree}
    A tree realisation \((T=(V,E), \Phi)\) of an \(n \times n\) distance matrix \(D\) is a {\emph{minimal tree realisation}} of \(D\), if there does not exist a proper subtree \(T_0=(V_0,E_0)\) of \(T\) with \(\Phi([n]) \subseteq V_0\).
\end{definition}
\medskip
Our next aim is to give a characterisation of minimal tree realisations of distance matrices (Proposition \ref{prop:minimaltree} below).
For this we need the following lemma.
\medskip
\begin{lemma}
    \label{lem:subtree}
    Let \(D\) be an \(n \times n\) distance matrix and \((T=(V,E), \Phi)\) be a tree realisation of \(D\).
    Assume \(T_0=(V_0, E_0)\) is a subtree of \(T\) with \(\Phi([n]) \subseteq V_0\), then \((T_0, \Phi)\) is also a tree realisation of \(D\). 
\end{lemma}

\begin{proof}
    Let \(\Phi(i) = v_i\) for \(i \in [n]\), and \(d_{T_0}\) and \(d_T\) be the distance functions of \(T_0\) and \(T\), respectively.
    It is sufficient to show that \[d_{T_0}(v_i, v_j) = d_{T}(v_i,v_j) \quad \forall i,j \in [n].\]
    Since \(T\) is a tree, there is a unique path of length \(d_{T}(v_i,v_j)\) between any two vertices \(v_i, v_j\) in \(T\).
    This path must also exist in \(T_0\), as otherwise \(T_0\) is disconnected, a contradiction.
\end{proof}
\medskip
As an immediate consequence of Lemma \ref{lem:subtree}, note that any tree realisable distance matrix \(D\) has also a minimal tree realisation.
\medskip
\begin{proposition}
    \label{prop:minimaltree}
    Let \(D\) be an \(n \times n\) distance matrix and \((T=(V,E), \Phi)\) be a tree realisation of \(D\).
    Let \(L \subseteq V\) be the set of leaves of \(T\). 
    Then \((T, \Phi)\) is a minimal tree realisation of \(D\) if and only if \(L \subseteq \Phi([n])\).
\end{proposition}

\begin{proof}
    Assume that \((T, \Phi)\) is a minimal tree realisation of \(D\) and that there is a leaf \(x \in L\) with \(x \not \in \Phi([n])\). 
    Let \(e \in E\) be the unique edge incident to \(x\).
    Then the proper subtree \(T_0=(V_0, E_0)\) with \(V_0 = V \setminus \{x\}\) and \(E_0 = E \setminus \{e\}\), together with \(\Phi: [n] \to V_0\) is also a tree realisation of \(D\) by Lemma \ref{lem:subtree}.
    This is a contradiction to \((T, \Phi)\) being minimal.

    Now let \((T, \Phi)\) be a tree realisation of \(D\) with \(L \subseteq \Phi([n])\).
    Assume that there exists a proper subtree \(T_0=(V_0, E_0)\) of \(T\) with \(\Phi([n]) \subseteq V_0\) which is also a tree realisation of \(D\).
    Then there exists at least one leaf \(x \in L\) with \(x \not \in V_0\).
    But this is a contradiction to the assumption that \(\Phi([n]) \subseteq V_0\).
\end{proof}
\medskip
By a result of Smolenskii 1962 \cite{Smolenskii1962}, it follows from our characterisation of minimal tree realisations that tree realisable distance matrices have a unique minimal tree realisation.
The following theorem by Zarecki\u{\i} 1965 \cite{Zareckii1965} provides a set of necessary and sufficient conditions for a distance matrix to have a tree realisation.
% (see \cite{Turner1970} for a (slightly incorrect) English translation; the theorem there is formulated as "[for an \(n \times n\) matrix, there exists a tree] with \(n\) endpoints (leaves)", which is too restrictive: the "\(n\) endpoints" should be replaced by a "set of \(n\) vertices \(x_1,...,x_n\) including all endpoints" as in the original statement in Zarecki\u{\i} 1965 \cite{Zareckii1965}).
\medskip
\begin{theorem}[See \cite{Zareckii1965}]
    \label{thm:zareckii1965}
    Let \(D\) be an \(n \times n\) matrix.
    Then \(D\) is a distance matrix and there exists a unique minimal tree realisation \((T=(V,E), \Phi)\) of \(D\) if and only if
    \begin{itemize}
        \item [(a)] For all \(i,j \in [n]\): \(D_{ij} \in \Z, D_{ij} = D_{ji} > 0\) for all \(i \neq j\), \(D_{ii} = 0\).
        \item [(b)] For all \(i,j,k \in [n]\): \(D_{ij} + D_{jk} - D_{ik}\) is even.
        \item [(c)] For all \(i,j,k,l \in [n]\): At least two of \(D_{ij} + D_{kl}, D_{ik} + D_{jl}, D_{il} + D_{jk}\) are equal and greater than or equal to the third.
    \end{itemize}
\end{theorem}
\medskip
Note for clarity, that condition (c) of Theorem \ref{thm:zareckii1965} is equivalent to
\[D_{ij} + D_{kl} \leq \max(D_{ik} + D_{jl}, D_{il} + D_{jk}),\]
for all \(i,j,k,l \in [n]\).

A similar result was discovered independently in the weighted graph context by Pereira 1969 \cite{SimoesPereira1969}, Buneman 1971 \cite{Buneman1971} and Buneman 1974 \cite{ Buneman1974} (see Semple and Steel 2003 \cite[Theorem 7.2.6]{Semple2003}).
% Weighted tree construction algorithms for real valued distance matrices are discussed in Semple and Steel 2003 \cite[Section 7.3]{Semple2003}.
Theorem \ref{thm:zareckii1965} gives a natural polynomial time algorithm with complexity \(O(n^4)\) to determine whether a given \(n \times n\) distance matrix has a tree realisation, thus solving \textsc{TreeCombDMR} with time complexity \(O(n^4)\).
However, this result does not provide such a tree realisation whenever it exists.
Zarecki\u{\i} 1965 \cite{Zareckii1965} gives another polynomial time algorithm to construct such a tree realisation of a distance matrix with time complexity \(O(n^4)\), but we will not discuss this algorithm in this paper.
Instead, we discuss an \(O(n^2)\) algorithm which solves \textsc{TreeCombDMR} and constructs a tree realisation of a distance matrix if it exists.
To do this, we utilise a connection with the minimum weighted tree realisation problem, as first described by Hakimi and Yau 1965 \cite{Hakimi1965DistanceMO}.
\medskip
\begin{problem}
    \label{prob:wt}
    \textsc{Minimum Weighted Tree Realisation Problem}\\
    \emph{Input:} A $n \times n$ matrix $D$ with non-negative real values. \\    
    \emph{Goal:} Produce a simple weighted tree $T=(V,E, w)$, \(V \geq n\), with minimum $\sum_{e \in E} w(e)$, where $w: E \rightarrow \R^+$,
     and an injective mapping $\Phi: [n] \rightarrow V$ such that the shortest-path distance function $d$ in $T$ satisfies
    \begin{equation*}
        \label{eq:wtr}
        d(\Phi(i), \Phi(j)) = D_{ij}
    \end{equation*}
        for all \(i,j \in [n]\).
        If no such tree exists, then output NO.
\end{problem}
\medskip
The minimum weighted tree realisation problem is known to be solvable in \(O(n^2)\) time by the algorithm of Culberson and Rudnicki 1989 \cite{Culberson1989} (which is similar to the algorithm of Batagelj et al. 1990 \cite{Batagelj1990}).
Furthermore, we have a simple criterion for a weighted tree realisation to be a minimum, given by the following theorem by Hakimi and Yau 1965 \cite[Theorems 3 and 4]{Hakimi1965DistanceMO}.
\medskip
\begin{theorem}[See \cite{Hakimi1965DistanceMO}]
    \label{thm:hy1965}
    If an \(n \times n\) distance matrix \(D\) has a weighted tree realisation \((T,\Phi, w)\) with no vertices in \(V \setminus \Phi([n])\) of degree less than or equal to 2, then \(T\) is a unique minimum weighted tree realisation of \(D\).
\end{theorem}
\medskip
Given an arbitrary weighted tree realisation \((T=(V,E, w), \Phi)\) of \(D\), we define the {\emph{canonical transformation}} of this tree to be as follows:
\begin{itemize}
    \item Remove all vertices \(v \in V \setminus \Phi([n])\) of degree 1, along with their incident edges.
    \item Successively, replace each vertex \(v \in V \setminus \Phi([n])\) of degree 2, along with its incident edges, by a single edge connecting its two neighbors, with weight equal to the sum of the weights of the two incident edges.
\end{itemize}
Theorem \ref{thm:hy1965} guarantees that the canonical transformation of any weighted tree realisation of a distance matrix \(D\) is the unique minimum.
Our aim is to show that a given \(n \times n\) distance matrix \(D\) is a YES-instance of \textsc{TreeCombDMR} if and only if the unique minimum weighted tree realisation of \(D\) exists and has integer valued edge weights.
Knowing this, we can simply apply the algorithm of Culberson and Rudnicki 1989 \cite{Culberson1989} to obtain the unique minimum weighted tree realisation of the input distance matrix and then check whether the edge weights are all integers.

\medskip
\begin{proposition}
    \label{prop:treeequaltr}
    Let \(D\) be an integer valued \(n \times n\) distance matrix.
    There exists a unique minimal tree realisation \((T=(V,E), \Phi)\) of \(D\)
    if and only if there exists a unique minimum weighted tree realisation \((T'=(V',E',w), \Phi)\) of \(D\) whose edge weights are all in \(\N\).
\end{proposition}

\begin{proof}
    Let \((T=(V,E), \Phi)\) be the unique minimal tree realisation of \(D\).
    Construct a weighted tree realisation \((T'=(V',E',w), \Phi)\) of \(D\) by assigning each edge \(e \in E\) a weight of~\(1\).
    This gives a weighted tree realisation of \(D\) with all edge weights in~\(\N\).
    We can then apply the canonical transformation to obtain the unique minimum weighted tree realisation of \(D\), as given by Theorem \ref{thm:hy1965}.
    As the edge weights in the initial weighted tree realisation were all in \(\N\), the edge weights in the unique minimum weighted tree realisation are also all in \(\N\) under this transformation.

    Now let \((T'=(V',E',w), \Phi)\) be the unique minimum weighted tree realisation of~\(D\) with all edge weights in \(\N\).
    Construct the simple tree realisation \((T=(V,E), \Phi)\) of~\(D\) by replacing each edge \(e \in E'\) with weight \(w(e)\) by an elementary path of length \(w(e)\) between its endpoints.
    This is a tree realisation of \(D\) as the shortest path distance function of \(T\) coincides with the distance matrix \(D\).
    As we have a tree realisation of~ \(D\), Lemma \ref{lem:subtree} implies that \(D\) also admits a minimal tree realisation.
    By Smolenskii 1962~\cite{Smolenskii1962}, this minimal tree realisation is unique.
\end{proof}
\medskip

If follows from Proposition \ref{prop:treeequaltr} that we can solve \textsc{TreeCombDMR} in \(O(n^2)\) time by the following algorithm.
\medskip
\begin{alg} [Solving \textsc{TreeCombDMR}]{\phantom{Bob}}
    \label{alg:cr1989Ours}
    \begin{itemize}
        \item[] {\bf{Input}}: An integer valued \(n \times n\) distance matrix \(D\).
        \item[] {\bf{Output}}: A tree realisation \((T=(V,E), \Phi)\) of \(D\) or a statement that no such tree realisation exists.
        \item[1.] Employ the algorithm of Culberson and Rudnicki 1989 \cite{Culberson1989} to obtain the unique minimum weighted tree realisation \((T'=(V',E',w), \Phi)\) of \(D\) with \(w: E' \to \R^+\) if it exists.
        Otherwise, return that no such tree realisation exists (by Proposition \ref{prop:treeequaltr}).
        \item[2.] If there exists an edge \(e \in E'\) with a non-integer weight \(w(e)\), return that no such tree realisation exists (by Proposition~\ref{prop:treeequaltr}).
        \item[3.] Construct the tree realisation \((T=(V,E), \Phi)\) of \(D\) by replacing each edge \(e \in E'\) with weight \(w(e) \in \N\) by an elementary path of length \(w(e)\) between its endpoints.
        \item[4.] Return the tree realisation \((T=(V,E), \Phi)\).
    \end{itemize}
\end{alg}
\medskip

\section{Conclusions} \label{sec:conclusions}

In this paper we introduced and studied the computational complexity of the problem
\textsc{$k$-Combinatorial Distance Matrix Realisation ($k$-CombDMR)}. 
Our main result is that \textsc{$k$-CombDMR} is NP-complete for every \(k \geq 3\).
We also provided polynomial time algorithms to solve \(0\)-\textsc{CombDMR}, \(1\)-\textsc{CombDMR} and \(2\)-\textsc{CombDMR}, which also produce a graph realisation of the distance matrix if it exists.
Furthermore, for tree realisations of a distance matrix, we showed that the problem \textsc{TreeCombDMR} is solvable in \(O(n^2)\) time and we provided an algorithm for constructing a tree realisation if it exists.

Possible future work includes (1) further constraints to the desired graph realisation, such as planarity, bi-partiteness, or other graph properties, and
(2) potential extensions of our results to directed graphs.
We could also (3) seek approximation algorithms for \textsc{$k$-CombDMR} or prove approximation hardness. 
Finally, we could (4) investigate the complexity of the weak version of \textsc{$k$-CombDMR} (as in Chung, Garrett and Graham 2001 \cite{Chung2001}), where the entries of the input distance matrix are simply lower bounds on the shortest path distances between vertices in the desired graph realisation.

\section{Acknowledgments}
We thank Magnus Bordewich for insightful comments regarding phylogenetics. 
Additionally, we extend our thanks to Tharsus Limited for the support they have provided to enable this research.

% \section{Section title of first appendix}\label{secA1}

% An appendix contains supplementary information that is not an essential part of the text itself but which may be helpful in providing a more comprehensive understanding of the research problem or it is information that is too cumbersome to be included in the body of the paper.




% \end{appendices}


\bibliography{sn-bibliography}% common bib file


\begin{thebibliography}{20}
\ifx \bisbn   \undefined \def \bisbn  #1{ISBN #1}\fi
\ifx \binits  \undefined \def \binits#1{#1}\fi
\ifx \bauthor  \undefined \def \bauthor#1{#1}\fi
\ifx \batitle  \undefined \def \batitle#1{#1}\fi
\ifx \bjtitle  \undefined \def \bjtitle#1{#1}\fi
\ifx \bvolume  \undefined \def \bvolume#1{\textbf{#1}}\fi
\ifx \byear  \undefined \def \byear#1{#1}\fi
\ifx \bissue  \undefined \def \bissue#1{#1}\fi
\ifx \bfpage  \undefined \def \bfpage#1{#1}\fi
\ifx \blpage  \undefined \def \blpage #1{#1}\fi
\ifx \burl  \undefined \def \burl#1{\textsf{#1}}\fi
\ifx \doiurl  \undefined \def \doiurl#1{\url{https://doi.org/#1}}\fi
\ifx \betal  \undefined \def \betal{\textit{et al.}}\fi
\ifx \binstitute  \undefined \def \binstitute#1{#1}\fi
\ifx \binstitutionaled  \undefined \def \binstitutionaled#1{#1}\fi
\ifx \bctitle  \undefined \def \bctitle#1{#1}\fi
\ifx \beditor  \undefined \def \beditor#1{#1}\fi
\ifx \bpublisher  \undefined \def \bpublisher#1{#1}\fi
\ifx \bbtitle  \undefined \def \bbtitle#1{#1}\fi
\ifx \bedition  \undefined \def \bedition#1{#1}\fi
\ifx \bseriesno  \undefined \def \bseriesno#1{#1}\fi
\ifx \blocation  \undefined \def \blocation#1{#1}\fi
\ifx \bsertitle  \undefined \def \bsertitle#1{#1}\fi
\ifx \bsnm \undefined \def \bsnm#1{#1}\fi
\ifx \bsuffix \undefined \def \bsuffix#1{#1}\fi
\ifx \bparticle \undefined \def \bparticle#1{#1}\fi
\ifx \barticle \undefined \def \barticle#1{#1}\fi
\bibcommenthead
\ifx \bconfdate \undefined \def \bconfdate #1{#1}\fi
\ifx \botherref \undefined \def \botherref #1{#1}\fi
\ifx \url \undefined \def \url#1{\textsf{#1}}\fi
\ifx \bchapter \undefined \def \bchapter#1{#1}\fi
\ifx \bbook \undefined \def \bbook#1{#1}\fi
\ifx \bcomment \undefined \def \bcomment#1{#1}\fi
\ifx \oauthor \undefined \def \oauthor#1{#1}\fi
\ifx \citeauthoryear \undefined \def \citeauthoryear#1{#1}\fi
\ifx \endbibitem  \undefined \def \endbibitem {}\fi
\ifx \bconflocation  \undefined \def \bconflocation#1{#1}\fi
\ifx \arxivurl  \undefined \def \arxivurl#1{\textsf{#1}}\fi
\csname PreBibitemsHook\endcsname

%%% 1
\bibitem[\protect\citeauthoryear{Semple and Steel}{2003}]{Semple2003}
\begin{bbook}
\bauthor{\bsnm{Semple}, \binits{C.}},
\bauthor{\bsnm{Steel}, \binits{M.}}:
\bbtitle{Phylogenetics}.
\bsertitle{Oxford Lecture Series in Mathematics and its Applications},
vol. \bseriesno{24}.
\bpublisher{Oxford University Press},
\blocation{Oxford}
(\byear{2003})
\end{bbook}
\endbibitem

%%% 2
\bibitem[\protect\citeauthoryear{Chung et~al.}{2001}]{Chung2001}
\begin{barticle}
\bauthor{\bsnm{Chung}, \binits{F.}},
\bauthor{\bsnm{Garrett}, \binits{M.}},
\bauthor{\bsnm{Graham}, \binits{R.}},
\bauthor{\bsnm{Shallcross}, \binits{D.}}:
\batitle{Distance realization problems with applications to internet
  tomography}.
\bjtitle{Journal of Computer and System Sciences}
\bvolume{63}(\bissue{3}),
\bfpage{432}--\blpage{448}
(\byear{2001})
\end{barticle}
\endbibitem

%%% 3
\bibitem[\protect\citeauthoryear{Herman and Kuba}{2012}]{herman2012discrete}
\begin{bbook}
\bauthor{\bsnm{Herman}, \binits{G.T.}},
\bauthor{\bsnm{Kuba}, \binits{A.}}:
\bbtitle{Discrete Tomography: Foundations, Algorithms, and Applications}.
\bpublisher{Springer},
\blocation{New York City}
(\byear{2012})
\end{bbook}
\endbibitem

%%% 4
\bibitem[\protect\citeauthoryear{Bar-Noy et~al.}{2021}]{Barnoy2021}
\begin{bchapter}
\bauthor{\bsnm{Bar-Noy}, \binits{A.}},
\bauthor{\bsnm{B{\"o}hnlein}, \binits{T.}},
\bauthor{\bsnm{Peleg}, \binits{D.}},
\bauthor{\bsnm{Perry}, \binits{M.}},
\bauthor{\bsnm{Rawitz}, \binits{D.}}:
\bctitle{Relaxed and approximate graph realizations}.
In: \bbtitle{International Workshop on Combinatorial Algorithms},
pp. \bfpage{3}--\blpage{19}
(\byear{2021})
\end{bchapter}
\endbibitem

%%% 5
\bibitem[\protect\citeauthoryear{Hakimi and Yau}{1965}]{Hakimi1965DistanceMO}
\begin{barticle}
\bauthor{\bsnm{Hakimi}, \binits{S.L.}},
\bauthor{\bsnm{Yau}, \binits{S.S.}}:
\batitle{Distance matrix of a graph and its realizability}.
\bjtitle{Quarterly of Applied Mathematics}
\bvolume{22},
\bfpage{305}--\blpage{317}
(\byear{1965})
\end{barticle}
\endbibitem

%%% 6
\bibitem[\protect\citeauthoryear{Dress}{1984}]{Dress1984}
\begin{barticle}
\bauthor{\bsnm{Dress}, \binits{A.W.M.}}:
\batitle{Trees, tight extensions of metric spaces, and the cohomological
  dimension of certain groups: a note on combinatorial properties of metric
  spaces}.
\bjtitle{Adv. in Math.}
\bvolume{53}(\bissue{3}),
\bfpage{321}--\blpage{402}
(\byear{1984})
\end{barticle}
\endbibitem

%%% 7
\bibitem[\protect\citeauthoryear{Alth\"{o}fer}{1988}]{Althofer1988}
\begin{barticle}
\bauthor{\bsnm{Alth\"{o}fer}, \binits{I.}}:
\batitle{On optimal realizations of finite metric spaces by graphs}.
\bjtitle{Discrete Comput. Geom.}
\bvolume{3}(\bissue{2}),
\bfpage{103}--\blpage{122}
(\byear{1988})
\end{barticle}
\endbibitem

%%% 8
\bibitem[\protect\citeauthoryear{Culberson and Rudnicki}{1989}]{Culberson1989}
\begin{barticle}
\bauthor{\bsnm{Culberson}, \binits{J.C.}},
\bauthor{\bsnm{Rudnicki}, \binits{P.}}:
\batitle{A fast algorithm for constructing trees from distance matrices}.
\bjtitle{Information Processing Letters}
\bvolume{30}(\bissue{4}),
\bfpage{215}--\blpage{220}
(\byear{1989})
\end{barticle}
\endbibitem

%%% 9
\bibitem[\protect\citeauthoryear{Cormen et~al.}{1990}]{Cormen1990}
\begin{bbook}
\bauthor{\bsnm{Cormen}, \binits{T.H.}},
\bauthor{\bsnm{Leiserson}, \binits{C.E.}},
\bauthor{\bsnm{Rivest}, \binits{R.L.}}:
\bbtitle{Introduction to Algorithms},
\bedition{1st} edn.
\bpublisher{MIT Press and McGraw-Hill},
\blocation{Cambridge, Massachusetts}
(\byear{1990})
\end{bbook}
\endbibitem

%%% 10
\bibitem[\protect\citeauthoryear{Papadimitriou}{1994}]{papadimitriou1994computational}
\begin{bbook}
\bauthor{\bsnm{Papadimitriou}, \binits{C.H.}}:
\bbtitle{Computational Complexity}.
\bpublisher{Addison-Wesley},
\blocation{Reading, Massachusetts}
(\byear{1994})
\end{bbook}
\endbibitem

%%% 11
\bibitem[\protect\citeauthoryear{Krom}{1967}]{Krom1967}
\begin{barticle}
\bauthor{\bsnm{Krom}, \binits{M.R.}}:
\batitle{The decision problem for a class of first-order formulas in which all
  disjunctions are binary}.
\bjtitle{Mathematical Logic Quarterly}
\bvolume{13}(\bissue{1-2}),
\bfpage{15}--\blpage{20}
(\byear{1967})
\end{barticle}
\endbibitem

%%% 12
\bibitem[\protect\citeauthoryear{Karp}{1972}]{Karp1972}
\begin{bchapter}
\bauthor{\bsnm{Karp}, \binits{R.M.}}:
\bctitle{Reducibility among combinatorial problems}.
In: \bbtitle{Complexity of Computer Computations ({P}roc. {S}ympos., {IBM}
  {T}homas {J}. {W}atson {R}es. {C}enter, {Y}orktown {H}eights, {N}.{Y}.,
  1972)}.
\bsertitle{The IBM Research Symposia Series},
pp. \bfpage{85}--\blpage{103}.
\bpublisher{Plenum},
\blocation{New York-London}
(\byear{1972})
\end{bchapter}
\endbibitem

%%% 13
\bibitem[\protect\citeauthoryear{Lov\'{a}sz}{1973}]{Lovasz1973}
\begin{bchapter}
\bauthor{\bsnm{Lov\'{a}sz}, \binits{L.}}:
\bctitle{Coverings and coloring of hypergraphs}.
In: \bbtitle{Proceedings of the {F}ourth {S}outheastern {C}onference on
  {C}ombinatorics, {G}raph {T}heory and {C}omputing ({F}lorida {A}tlantic
  {U}niv., {B}oca {R}aton, {F}la., 1973)}.
\bsertitle{Congress. Numer.},
vol. \bseriesno{VIII},
pp. \bfpage{3}--\blpage{12}.
\bpublisher{Utilitas Math.},
\blocation{Winnipeg, MB, USA}
(\byear{1973})
\end{bchapter}
\endbibitem

%%% 14
\bibitem[\protect\citeauthoryear{Stockmeyer}{1973}]{Stockmeyer1973}
\begin{barticle}
\bauthor{\bsnm{Stockmeyer}, \binits{L.}}:
\batitle{Planar 3-colorability is polynomial complete}.
\bjtitle{SIGACT News}
\bvolume{5}(\bissue{3}),
\bfpage{19}--\blpage{25}
(\byear{1973})
\end{barticle}
\endbibitem

%%% 15
\bibitem[\protect\citeauthoryear{Zarecki\u{\i}}{1965}]{Zareckii1965}
\begin{barticle}
\bauthor{\bsnm{Zarecki\u{\i}}, \binits{K.A.}}:
\batitle{Constructing a tree on the basis of a set of distances between the
  hanging vertices}.
\bjtitle{Uspehi Mat. Nauk}
\bvolume{20}(\bissue{6}),
\bfpage{90}--\blpage{92}
(\byear{1965})
\end{barticle}
\endbibitem

%%% 16
\bibitem[\protect\citeauthoryear{Smolenskii}{1962}]{Smolenskii1962}
\begin{barticle}
\bauthor{\bsnm{Smolenskii}, \binits{E.A.}}:
\batitle{A method for the linear recording of graphs}.
\bjtitle{Zhurnal Vychislitel'noi Matematiki i Matematicheskoi Fiziki}
\bvolume{3}(\bissue{2}),
\bfpage{371}--\blpage{372}
(\byear{1962}).
\bcomment{Trans: USSR Computational Mathematics and Mathematical Physics, Vol.
  3, No. 2, pp. 396-397 (1962); Abs: Mathematical Reviews, Vol. 32, p. 1270;
  Abs: Referativnyi Zhurnal Matematika, No. 7V290 (1963)}
\end{barticle}
\endbibitem

%%% 17
\bibitem[\protect\citeauthoryear{Sim\~{o}es Pereira}{1969}]{SimoesPereira1969}
\begin{barticle}
\bauthor{\bsnm{Pereira}, \binits{J.M.S.}}:
\batitle{A note on the tree realizability of a distance matrix}.
\bjtitle{J. Combinatorial Theory}
\bvolume{6},
\bfpage{303}--\blpage{310}
(\byear{1969})
\end{barticle}
\endbibitem

%%% 18
\bibitem[\protect\citeauthoryear{Buneman}{1971}]{Buneman1971}
\begin{bchapter}
\bauthor{\bsnm{Buneman}, \binits{P.}}:
\bctitle{The recovery of trees from measures of dissimilarity}.
In: \bbtitle{Mathematics the the Archeological and Historical Sciences},
pp. \bfpage{387}--\blpage{395}.
\bpublisher{Edinburgh University Press},
\blocation{United Kingdom}
(\byear{1971})
\end{bchapter}
\endbibitem

%%% 19
\bibitem[\protect\citeauthoryear{Buneman}{1974}]{Buneman1974}
\begin{barticle}
\bauthor{\bsnm{Buneman}, \binits{P.}}:
\batitle{A note on the metric properties of trees}.
\bjtitle{J. Combinatorial Theory Ser. B}
\bvolume{17},
\bfpage{48}--\blpage{50}
(\byear{1974})
\end{barticle}
\endbibitem

%%% 20
\bibitem[\protect\citeauthoryear{Batagelj et~al.}{1990}]{Batagelj1990}
\begin{barticle}
\bauthor{\bsnm{Batagelj}, \binits{V.}},
\bauthor{\bsnm{Pisanski}, \binits{T.}},
\bauthor{\bsnm{Simões-Pereira}, \binits{J.}}:
\batitle{An algorithm for tree-realizability of distance matrices}.
\bjtitle{International Journal of Computer Mathematics}
\bvolume{34},
\bfpage{171}--\blpage{176}
(\byear{1990})
\end{barticle}
\endbibitem

\end{thebibliography}

\end{document}